\documentclass[]{article}
\title{Copula Modelling of Serially Correlated Multivariate Data with Hidden Structures}
\author{Robert Zimmerman, Radu V. Craiu and Vianey Leos-Barajas}
\date{\today}

\usepackage{float}
\usepackage{amsmath}
\usepackage{amsthm}
\usepackage{thmtools}
\usepackage{array}
\usepackage{booktabs}
\usepackage{amsfonts}
\usepackage{amssymb}
\usepackage{amsxtra}
\usepackage{enumerate}
\usepackage{bm}
\usepackage{bbm}
\usepackage{color}
\usepackage{mathtools}
\usepackage{tikz}
\usepackage{dcolumn}
\usepackage{graphicx}
\usepackage{natbib}
\usepackage[colorlinks = true, linkcolor = blue, urlcolor = blue, citecolor = blue]{hyperref}

\makeatletter
\let\save@mathaccent\mathaccent
\newcommand*\if@single[3]{%
  \setbox0\hbox{${\mathaccent"0362{#1}}^H$}%
  \setbox2\hbox{${\mathaccent"0362{\kern0pt#1}}^H$}%
  \ifdim\ht0=\ht2 #3\else #2\fi
  }
\newcommand*\rel@kern[1]{\kern#1\dimexpr\macc@kerna}
\newcommand*\widebar[1]{\@ifnextchar^{{\wide@bar{#1}{0}}}{\wide@bar{#1}{1}}}
\newcommand*\wide@bar[2]{\if@single{#1}{\wide@bar@{#1}{#2}{1}}{\wide@bar@{#1}{#2}{2}}}
\newcommand*\wide@bar@[3]{%
  \begingroup
  \def\mathaccent##1##2{%
    \let\mathaccent\save@mathaccent
    \if#32 \let\macc@nucleus\first@char \fi
    \setbox\z@\hbox{$\macc@style{\macc@nucleus}_{}$}%
    \setbox\tw@\hbox{$\macc@style{\macc@nucleus}{}_{}$}%
    \dimen@\wd\tw@
    \advance\dimen@-\wd\z@
    \divide\dimen@ 3
    \@tempdima\wd\tw@
    \advance\@tempdima-\scriptspace
    \divide\@tempdima 10
    \advance\dimen@-\@tempdima
    \ifdim\dimen@>\z@ \dimen@0pt\fi
    \rel@kern{0.6}\kern-\dimen@
    \if#31
      \overline{\rel@kern{-0.6}\kern\dimen@\macc@nucleus\rel@kern{0.4}\kern\dimen@}%
      \advance\dimen@0.4\dimexpr\macc@kerna
      \let\final@kern#2%
      \ifdim\dimen@<\z@ \let\final@kern1\fi
      \if\final@kern1 \kern-\dimen@\fi
    \else
      \overline{\rel@kern{-0.6}\kern\dimen@#1}%
    \fi
  }%
  \macc@depth\@ne
  \let\math@bgroup\@empty \let\math@egroup\macc@set@skewchar
  \mathsurround\z@ \frozen@everymath{\mathgroup\macc@group\relax}%
  \macc@set@skewchar\relax
  \let\mathaccentV\macc@nested@a
  \if#31
    \macc@nested@a\relax111{#1}%
  \else
    \def\gobble@till@marker##1\endmarker{}%
    \futurelet\first@char\gobble@till@marker#1\endmarker
    \ifcat\noexpand\first@char A\else
      \def\first@char{}%
    \fi
    \macc@nested@a\relax111{\first@char}%
  \fi
  \endgroup
}
\makeatother
\renewcommand{\bar}{\widebar}

\newcommand{\bvzero}{\mathbf{0}}

\newcommand{\bvg}{\mathbf{g}}

\newcommand{\bva}{\mathbf{a}}
\newcommand{\bvu}{\mathbf{u}}
\newcommand{\bvv}{\mathbf{v}}
\newcommand{\bvw}{\mathbf{w}}

\newcommand{\bvX}{\mathbf{X}}
\newcommand{\bvy}{\mathbf{y}}
\newcommand{\bvA}{\mathbf{A}}
\newcommand{\bvD}{\mathbf{D}}
\newcommand{\bvF}{\mathbf{F}}
\newcommand{\bvG}{\mathbf{G}}
\newcommand{\bvH}{\mathbf{H}}
\newcommand{\bvJ}{\mathbf{J}}
\newcommand{\bvL}{\mathbf{L}}
\newcommand{\bvT}{\mathbf{T}}
\newcommand{\bvV}{\mathbf{V}}
\newcommand{\bvU}{\mathbf{U}}

\newcommand{\bvY}{\mathbf{Y}}

\newcommand{\bveta}{\bm{\eta}}
\newcommand{\bvgamma}{\bm{\gamma}}

\newcommand{\bvphi}{\bm{\phi}}
\newcommand{\bvpi}{\bm{\pi}}
\newcommand{\bvpsi}{\bm{\psi}}
\newcommand{\bvxi}{\bm{\xi}}
\newcommand{\bvtheta}{\bm{\theta}}

\newcommand{\bvlambda}{\bm{\lambda}}

\newcommand{\bvGamma}{\bm{\Gamma}}
\newcommand{\bvTheta}{\bm{\Theta}}

\newcommand{\norm}[1]{\left\lVert#1\right\rVert}
\newcommand{\bb}[1]{\mathbb{#1}}
\newcommand{\sA}{\mathcal{A}}

\newcommand{\sN}{\mathcal{N}}

\newcommand{\sQ}{\mathcal{Q}}
\newcommand{\sS}{\mathcal{S}}

\newcommand{\sX}{\mathcal{X}}
\newcommand{\sY}{\mathcal{Y}}

\newcommand{\R}{\bb{R}}
\newcommand{\C}{\bb{C}}

\newcommand{\Prb}[1]{\bb{P} \left( #1 \right)}
\newcommand{\parPrb}[2]{\bb{P}_{#1} \left( #2 \right)}
\newcommand{\E}[1]{\bb{E} \left[ #1 \right]}

\newcommand{\one}[1]{\mathbbm{1}_{#1}}
\newcommand{\onee}[1]{\mathbbm{1}\left({#1}\right)}

\renewcommand{\exp}[1]{\hspace{1pt}\mathrm{exp}\left(#1\right)}
\newcommand{\llog}[1]{\hspace{1pt}\mathrm{log}\left(#1\right)}

\newcommand{\inD}{\stackrel{d}{\longrightarrow}}
\newcommand{\iid}{\stackrel{iid}{\sim}}

\newcommand{\argmax}[1]{\underset{#1}{\mathrm{argmax}}\,}

\newcommand{\X}{\bvX}

\newcommand{\sumin}{\sum_{i=1}^n}

\newcommand{\psiinv}{\psi_\theta^{-1}}
\newcommand{\oneT}{_{1:T}}

\newcommand{\hX}{\hat{X}}
\newcommand{\sumtT}{\sum_{t=1}^T}
\newcommand{\sumjK}{\sum_{j=1}^K}
\newcommand{\sumkK}{\sum_{k=1}^K}

\newcommand{\dif}{\mathop{}\!\mathrm{d}}

\DeclareMathOperator*{\argsup}{arg\,sup}

\newcommand\indep{\protect\mathpalette{\protect\independenT}{\perp}}
\def\independenT#1#2{\mathrel{\rlap{$#1#2$}\mkern2mu{#1#2}}}

\theoremstyle{plain}
\newtheorem{theorem}{Theorem}[section]
\newtheorem*{theorem*}{Theorem}
\newtheorem*{definition*}{Definition}
\newtheorem{lemma}{Lemma}[section]
\newtheorem{corollary}{Corollary}[section]
\newtheorem{proposition}{Proposition}[section]

\theoremstyle{remark}
\newtheorem*{remark}{Remark}

\begin{document}

\maketitle

\begin{abstract}
We propose a copula-based extension of the hidden Markov model (HMM) which applies when the observations recorded at each time in the sample are multivariate. The joint model produced by the copula extension allows  decoding of the hidden states based on information from multiple observations. However, unlike the case of independent marginals, the copula dependence structure embedded into the likelihood poses additional computational challenges. We tackle the latter using a theoretically-justified variation of the EM algorithm developed within the framework of inference functions for margins. We illustrate the method using numerical experiments and an analysis of house occupancy. 
\end{abstract}


\section{Introduction}

Statistical models aim to capture the generating mechanism of the observed data. In many instances, one must consider various forms of {\it dark data} \citep{hand2020dark}, such as missing data that could not be observed at all or that has been modified during the sampling stage, or data that is unobservable but can be injected into the model as latent variables to add meaning and to enhance the interpretability of the statistical model. The latter approach is used to set up a hidden Markov model (HMM) in which the observed data consists of serially correlated observations on each item. The latent aspect of the model is a Markov chain with a discrete state space that evolves on the same time scale as the observed processes. The HMM postulates that the distribution of the observed data depends on the state of the Markov chain at that same time. Estimating the hidden structure can often illuminate the underlying workings of the system \citep{rabiner1986introduction,zucchini2017hidden}; in some cases, the Markov latent structure can be associated to a real hidden mechanism with some certainty. Examples of HMMs abound in numerous domains including economics \citep{du2020using}, medicine \citep{williams2020bayesian}, ecology \citep{ailliot2009space}, and sports \citep{otting2021copula}. 

Several distinct problems can be tackled using the HMM formulation. If one aims to identify the number of states that could have generated the data, and subsequently to \emph{decode} the state sequence, an unsupervised approach relies on clustering the observed features and assigning each to the state most likely to have generated it. On the other hand, the supervised version of the analysis implies that enough items have been monitored and their hidden states labeled to constitute a complete set of training data. In this case, one aims to predict the hidden states for observations outside the training sample. In either the supervised or unsupervised setting, one might also aim to impute or predict missing values from the observed features in a dataset that exhibits missing patterns.

Our work is motivated by situations in which the observations recorded at each time for each item are multivariate, and one must decide how to \emph{jointly} integrate the information they contain about the hidden states. When the observed data in an HMM setup is vector-valued, the state-dependent distributions can be constructed by assuming either contemporaneous conditional independence or longitudinal conditional independence \citep{zucchini2017hidden}. In the former case, the multivariate state-dependent distribution for the vector-valued observations is constructed as a product of marginal distributions, while in the latter case it is typically assumed to be multivariate Gaussian, which is highly tractable and easily implemented within the standard array of HMM algorithms \citep{CappeHMM, zucchini2017hidden}. However, the data which the multivariate Gaussian distribution can model is limited; in the bivariate case, for example, it cannot capture dependence in the extremes of the upper-right or lower-left tails (a property known as \emph{tail dependence} \citep{embrechts2001modelling}). One of our main aims here is to propose a general approach for integration of the continuous-valued information provided by multivariate observations in HMM models. 

Copulas have become a ubiquitous tool in modelling complex dependence structures. \cite{Sklar:1959} provided a theoretical foundation and demonstrated that any multivariate distribution can be represented by its marginal distributions and a copula that fully describes their interdependence. Moreover,  this decomposition is unique whenever the marginal distributions are continuous. Copulas have been widely used in analysis of survival data \citep{goet2, Hougaard:2000} and in various statistical applications, including actuarial science \citep{Valdez:1998}, hydrology \citep{Genest:2007if, Dupuis:2007qo} and finance \citep{cherubini2004copula, da2012modeling}, to name only a few. The developments of conditional copulas \citep{acar2011dependence, vog, gam-cc, cra-sabeti, levi2018bayesian} and vine copulas \citep{czado,min} have further expanded the range of dependence structures available for statistical modelling. 

The paper makes several contributions. First, we extend the modelling toolbox for HMMs with multivariate observations by considering a copula model for the distribution of observed data. Our numerical experiments show that when the copula is allowed to vary with the hidden state variable, the accuracy of hidden state identification increases. Because using copulas to jointly model the observed data presents additional computational challenges compared to the case of independent marginals or multivariate Gaussians, our second main contribution is to develop a new optimization procedure in which we integrate the \emph{inference functions for margins (IFM)} method of \cite{joe1996estimation}
within the \emph{ES algorithm} of \cite{elashoff2004algorithm}.

In the next section we introduce the model along with numerical evidence of the gains in decoding precision when using the copula-based joint modelling of the observed data. The computational methods used for estimation are introduced and theoretically justified in Sections \ref{sec:parameterestimation} and \ref{sec:TheoreticalResults}. Section \ref{sec:simsandapps} contains numerical experiments based on simulations as well as analysis of house occupancy data. The paper ends with a summary and a discussion of future work.

\section{A Copula-Based HMM}

\subsection{The General Model}\label{sub:genmod}

Let $X_1, X_2, \ldots$ be an unobserved discrete-time first-order Markov process taking values in a finite state space $\sX = \{1, 2, \ldots, K\}$ with initial distribution $\bvpi:\sQ \to [0,1]$ and transition probability matrix $\bvGamma = [\gamma_{i,j}]_{i,j \in \sX}$. Let $\bvY_1, \bvY_2, \ldots \bvY_T \in \R^d$ represent our observed data, assumed to satisfy the conditional independence structure $(\bvY_i \mid X_i) \indep (\bvY_j \mid X_j)$ for $i \neq j$.

\definecolor{lightblue}{RGB}{225, 232, 239}
\begin{figure}[ht]
  \centering
  \scalebox{0.85}{
    
\begin{tikzpicture}[shorten >=1pt,node distance=3cm,auto]
  \tikzstyle{state}=[shape=circle,thick,draw, minimum size = 1.5cm, fill=lightblue]
  \tikzstyle{obs}=[shape=circle,thick,draw, minimum size = 1.5cm]

  \node (Cold) {$\cdots$};
  \node[state, right of=Cold] (Ctm2) {$X_{t-2}$};
  \node[state, right of=Ctm2] (Ctm1) {$X_{t-1}$};
  \node[state, right of=Ctm1] (Ct) {$X_{t}$};
  \node[state, right of=Ct] (Ctp1) {$X_{t+1}$};
  \node[state, right of=Ctp1] (Ctp2) {$X_{t+2}$};
  \node[right of=Ctp2] (Cnew) {$\cdots$};

  \node[obs, below of=Ctm2] (Xtm2) {$\bvY_{t-2}$};
  \node[obs, below of=Ctm1] (Xtm1) {$\bvY_{t-1}$};
  \node[obs, below of=Ct] (Xt) {$\bvY_{t}$};
  \node[obs, below of=Ctp1] (Xtp1) {$\bvY_{t+1}$};
  \node[obs, below of=Ctp2] (Xtp2) {$\bvY_{t+2}$};

  \path[->, draw, thick, -latex] (Cold) -- (Ctm2);
  \path[->, draw, thick, -latex] (Ctm2) -- (Ctm1);
  \path[->, draw, thick, -latex] (Ctm1) -- (Ct);
  \path[->, draw, thick, -latex] (Ct) -- (Ctp1);
  \path[->, draw, thick, -latex] (Ctp1) -- (Ctp2);
  \path[->, draw, thick, -latex] (Ctp2) -- (Cnew);

  \path[->, draw, thick, -latex] (Ctm2) -- (Xtm2);
  \path[->, draw, thick, -latex] (Ctm1) -- (Xtm1);
  \path[->, draw, thick, -latex] (Ct) -- (Xt);
  \path[->, draw, thick, -latex] (Ctp1) -- (Xtp1);
  \path[->, draw, thick, -latex] (Ctp2) -- (Xtp2);
\end{tikzpicture}  } \caption{Standard HMM dependence structure} \label{fig:basicHMM}
\end{figure}
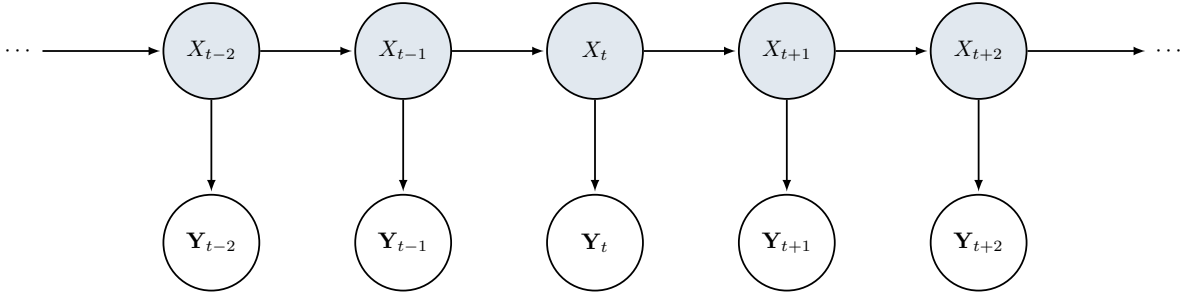

We assume that for each state $k \in \sX$, there exists a multivariate distribution $H_k(\cdot):\R^d \to [0,1]$ defined by $d$ univariate {\it $k$'th state-dependent marginal distributions}, $F_{k,1}(\cdot \,; \lambda_{k,1}), \ldots, F_{k,d}(\cdot \, ; \lambda_{k,d}): \R \to [0,1]$ and a $d$-dimensional parametric copula $C_k: [0,1]^d \to [0,1]$, such that
\begin{equation}\label{eq:state-dependentDF}
  \bvY_i \mid (X_i = k) \sim H_k(\cdot) = C_k\Big( F_{k,1}(\cdot \, ; \lambda_{k,1}), \ldots, F_{k,d}(\cdot \, ; \lambda_{k,d}) \, \Big| \, \theta_k\Big).
\end{equation}
This is the formulation of a generic finite-state time-homogeneous HMM (see Figure \ref{fig:basicHMM}). We further assume that each $C_k$ and $F_{k,h}$ admit densities $c_k$ and $f_{k,h}$ with respect to Lebesgue measure (on $\R^d$ and $\R$, respectively), implying the existence of the $k$'th state-dependent joint density
\[
    h_k(\bvy) 
    = c_k(F_{k,1}(y_1;\lambda_{k,1}),\ldots,F_{k,d}(y_d;\lambda_{k,1}) \mid \theta_k) \cdot \prod_{h=1}^d f_{k,h}(y_h;\lambda_{k,h}).
\]
Each $\lambda_{k,h} \in \R$ is a state- and margin-specific parameter, whose underlying parameter space depends on both the state $k$ and the component $h$. Thus, each state $k \in \sX$ generates a $d$-dimensional observation whose joint dependence structure is governed by a particular state-specific copula $C_k$ with associated parameter $\theta_k$. Note that while we assume that the $\lambda_{k,h}$'s and $\theta_k$'s are scalars for ease of presentation, our theory easily generalizes to multivariate parameters.

We further assume that each state-dependent copula $C_k$ is a member of a parametric family whose log-densities (and derivatives thereof) can be computed. The Gaussian copula is typically used to impose the same kind of dependence structure as the multivariate Gaussian distribution, while the $t$ copula is more suitable for capturing dependent extreme values \citep{demarta2005t}. \emph{Archimedean} copulas \citep{nelsen2007introduction} are also popularly used in statistical modelling. Each Archimedean copula can be explicitly represented as $C_\theta(\bvu) = \psi_\theta[ \psiinv(u_1) + \cdots + \psiinv(u_d)]$ for $\bvu \in [0,1]^d$, where the \emph{generator} $\psi_\theta:[0,\infty] \to [0,1]$ is a continuous decreasing function that satisfies $\psi_\theta(0) = 1$ and $\lim_{t \to \infty} \psi_\theta(t) = 0$, which is strictly decreasing on its support \citep{hofert2012likelihood}. The explicit functional representation of Archimedean copulas does not necessarily imply the analytic tractability of the usual quantities used for likelihood inference, but \cite{hofert2012likelihood} have derived explicit functional forms for the densities, generator derivatives, and score functions in the five ``classical'' one-parameter Archimedean copula families.

\subsection{Coupling Benefits: an Illustration} \label{sub:numericalillustration}

HMMs are generally used for two different purposes: to model a data-generating mechanism in which the Markov process may serve as a proxy for some process of interest, and to establish a classification of the hidden states. In applications where multivariate data regularly switch between different complex dependence structures and hierarchies, we can model the data-generating process as a finite-state HMM in which the state-dependent marginal distributions are linked together via copulas. The copula-within-HMM framework is highly flexible, because one has the freedom to vary both the copula itself as well as the marginal distributions between states (although as with HMMs in general, problems may arise if the state-dependent marginal distributions within each dimensions do not remain within the same parametric family). Such models have already been applied to financial data under the label of \emph{regime-switching copulas}; for example, \cite{nasri2020goodness} have used them for option pricing and have developed a goodness-of-fit approach to selecting the number of hidden states. In these applications, where the focus is on the data generating mechanism, the HMM is assessed not by its capacity to predict the states accurately, but rather via metrics that assess goodness-of-fit and the model's capacity to replicate key features of the data. 

A copula-within-HMM framework can also be used for the second aim, in which case the model would be assessed by its capacity to predict the hidden states. In this situation, a simple example shows that an  HMM performs poorly when assuming contemporaneous conditional independence of the observation process, even when the marginal state-dependent distributions are correctly specified. Consider bivariate data arising from a 2-state finite mixture model of length $T = 100$ with equal persistence between states (which may be viewed as a special case of an HMM), where the state-dependent distributions are Frank copulas \citep{genest1987frank} with standard normal margins, and the state-dependent copula parameters are $\theta_1 = -\theta$ and $\theta_2 = \theta$ for some fixed value of $\theta \in (0,100)$, representing extreme negative dependence and extreme positive dependence, respectively, when $|\theta|$ is large. In \autoref{fig:FrankZO}, we have plotted the zero-one loss for this model (i.e., $\ell_{01} = \frac{1}{T}\sum_{t=1}^T \one{\hat{X}_t \neq X_t}$, where $\hat{X}_t$ is the predicted state at time $t$) based on local state decoding as a function of $|\theta|$, as well as the zero-one loss based on an incorrectly-specified model with independent marginals. The zero-one loss for the independence model is, of course, constant; on the other hand, the zero-one loss for the true model appears to decay exponentially in $|\theta|$. 

\begin{figure}[ht]
    \centering
    \includegraphics[scale=0.15]{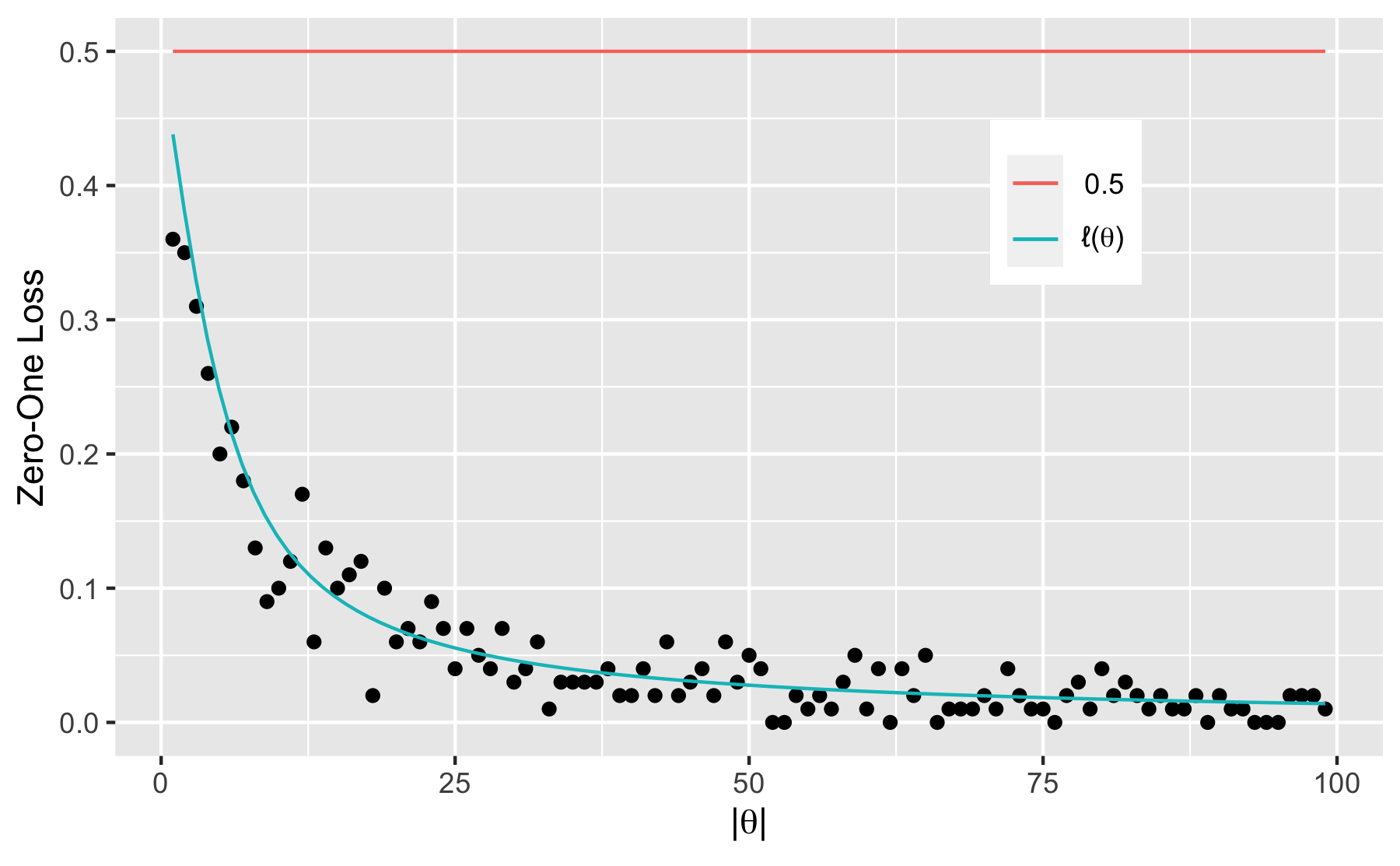}
    \caption{Empirical zero-one loss (black dots) for local decoding of a 2-state HMM with standard normal margins and Frank copulas with parameters $\theta$ and $-\theta$, for $|\theta| \in \{1, 2, \ldots, 100\}$. The red line is the zero-one loss when predictions are made assuming independent margins (which is constant in $\theta$); the blue line is the true expected zero-one loss at $\theta$.}
    \label{fig:FrankZO}
\end{figure}

In fact, the expected zero-one loss as a function of $\theta$ is given by $\ell_{01}(\theta) = \frac{1}{2} - \frac{2}{\theta}\llog{\cosh{\frac{\theta}{4}}}$. Such a measure can also be derived for other radially symmetric bivariate copula families within the same finite mixture model (see \autoref{cor:radsym01} within \autoref{app:proofs}); for example, under the same setup but with Gauss copulas in place of Frank copulas with state-dependent parameters $\rho_1 = -\rho$ and $\rho_2 = \rho \in (0,1)$, it can be shown that $\ell_{0,1}(\rho) = \mathrm{cos}^{-1}(\rho)/{\pi}$. These are, in fact, special cases of a far more general result:

\begin{theorem}\label{thm:ZOloss}
Let $\nu_{t,k} = \Prb{X_t = k} = [\bvpi \bvGamma^t]_k$. The expected zero-one loss of the classifications made by local decoding is given by
\begin{equation}\label{eq:ZOELapp}
\ell_{01}(\bveta) 
= \frac{1}{T}\sumtT \sumkK \nu_{t,k} \cdot \parPrb{\bveta}{\left.\frac{h_k(\bvY_t) \cdot \sum_{x_{-t} \in \sX^{T-1}} \kappa(x_{-t}, \bvY_{-t}) \cdot \gamma_{x_{t-1},k} \cdot \gamma_{k, x_{t+1}}}{\max_{j \neq k}\left\{ h_j(\bvY_t) \cdot \sum_{x_{-t} \in \sX^{T-1}} \kappa(x_{-t}, \bvY_{-t}) \cdot \gamma_{x_{t-1},j} \cdot \gamma_{j, x_{t+1}} \right\}} < 1 \right| X_t = k}
\end{equation}
where $\kappa(x_{-t}, \bvy_{-t}) = \pi_{x_1} \prod_{s \neq t} h_{x_s}(\bvy_s) \cdot \prod_{s \neq t,t+1} \gamma_{x_{s-1}, x_s}$. Moreover, if the HMM constitutes a finite mixture model, then \eqref{eq:ZOELapp} reduces to
\begin{equation}\label{eq:ZOELappFMM}
\ell_{01}(\bveta) 
= \frac{1}{T}\sumtT \sumkK \nu_{t,k} \int_{\R^d} \mathbbm{1}\left\{\frac{\nu_{t,k} \cdot h_k(\bvy_t)}{\max_{j \neq k} \nu_{t,j} \cdot h_j(\bvy_t)} < 1\right\} \dif H_k(\bvy_t).
\end{equation}
\end{theorem}

The complicated form of the HMM posterior state probabilities makes it challenging to derive useful bounds for the zero-one loss of the most general models; however, under certain conditions on the underlying Markov chain, one can deduce upper bounds for individual components of the loss function. For example, if the Markov chain spends most of its time travelling through state $k$, then the contributions to the loss function provided by the $k$'th state can be upper bounded by a more elegant expression:

\begin{proposition}\label{prop:probbound}
Suppose that $\parPrb{\bveta}{X_t = k \mid X_{t-1}, X_{t+1}} \geq \parPrb{\bveta}{X_t = j \mid X_{t-1}, X_{t+1}}$ for all $j \in \sX$. Then
\begin{align*}
&\phantom{{}={}}\parPrb{\bveta}{\left.\frac{h_k(\bvY_t) \cdot \sum_{x_{-t} \in \sX^{T-1}} \kappa(x_{-t}, \bvY_{-t}) \cdot \gamma_{x_{t-1},k} \cdot \gamma_{k, x_{t+1}}}{\max_{j \neq k}\left\{ h_j(\bvY_t) \cdot \sum_{x_{-t} \in \sX^{T-1}} \kappa(x_{-t}, \bvY_{-t}) \cdot \gamma_{x_{t-1},j} \cdot \gamma_{j, x_{t+1}} \right\}} < 1 \right| X_t = k} \\
&\leq \int_{\R^d} \mathbbm{1}\left\{\frac{h_k(\bvy_t)}{\max_{j \neq k} h_j(\bvy_t)} < 1\right\} \dif H_k(\bvy_t).
\end{align*}
\end{proposition}

The effect of the copulas in the HMM can be analyzed more carefully through the lens of the integral in \autoref{prop:probbound}; we have the following result for the bivariate case that explicitly relates the integral to the strength of the dependence within any particular state. Below, we write $\tau_k$ and $\rho_k$ for Kendall's $\tau$ and Spearman's $\rho$ of the copula in the $k$'th state, respectively.

\begin{proposition}\label{prop:ZO}
Let $d=2$, and fix $k \in \sX$. Then, as either  $|\tau_k| \to 1$ or $|\rho_k| \to 1$ while the other state-dependent copulas stay fixed, we have \[\int_{\R^d} \mathbbm{1}\left\{\frac{\omega_{k} \cdot h_k(\bvy)}{\max_{j \neq k} \omega_{j} \cdot h_j(\bvy)} < 1\right\} \dif H_k(\bvy) \to 0\] 
for any $(\omega_1, \ldots, \omega_K) \in \sS^{K-1}$ with $\omega_k > 0$.
\end{proposition}

In effect, \autoref{thm:ZOloss}, in combination with \autoref{prop:probbound} and \autoref{prop:ZO}, says that as the copula in a particularly ``common'' state approaches either of the Fr\'{e}chet-Hoeffding bounds, the observations produced by that state will become sufficiently distinguished for the local decoding algorithm to detect the state with complete accuracy, thereby eliminating its contribution to the zero-one loss. Proofs of these results are provided in \autoref{app:proofs}.

\section{Parameter Estimation}\label{sec:parameterestimation}

\subsection{Pitfalls of the Vanilla EM Algorithm}
\label{pitfall}

We consider a parametric approach for the copula specification, in which the parameter vector characterizes the state-dependent copulas and the marginal distributions. Parameter estimation for HMMs usually relies on the Baum-Welch algorithm \citep{baum1970maximization}, which is a particular instance of the  expectation-maximization (EM) algorithm \citep{dempster1977maximum}, with the sequence of hidden states playing the role of missing data. For the model proposed here, the E-step calculations are straightforward when all state-dependent densities can be evaluated at the current parameter estimates; alas, the maximization required to complete the M-step is considerably more difficult because an exact solution is usually unavailable in closed form, and one must then resort to numerical optimization. The task requires the maximization of a nonlinear function of all state-dependent parameters. Copula densities are particularly challenging to optimize directly due to the algebraic complexities of their associated score functions. In the case of Archimedean copulas, for example,  a naive symbolic differentiation of the log-likelihood function produces a number of terms that increase rapidly with the size of each trajectory \citep{hofert2012likelihood}; this renders the maximization of all parameters of the objective function computationally expensive. To illustrate, we present here a succinct version of the relevant mathematical equations; complete derivations may be found in \cite{zucchini2017hidden}, for example. 

The set of parameters in our model consists of the initial probability mass distribution $\{\pi_i\}_{i \in \sX}$ and the transition probabilities $\{\gamma_{i,j}\}_{i,j \in \sX}$ of the latent Markov chain, the parameters of the state-dependent marginal densities $\{\lambda_{i,h}\}_{i \in \sX, h \in [d]}$, and the copula parameters $\{\theta_i\}_{i \in \sX}$. We write $\bveta \in \R^p$ for the complete collection of these unknown parameters arranged in a vector, and $\bveta_k$ for the sub-vector of parameters associated with state $k \in \sX$ specifically. Our formulation here of the basic Baum-Welch algorithm follows \cite{zucchini2017hidden}. The observed data at time $t$ is denoted $\bvy_t$. Note that while our presentation deals with a single trajectory $\{(\bvy_t, X_t)\}_t$, the algorithm can easily be extended to handle multiple trajectories $\{\{(\bvy_{t,1}, X_{t,1})\}_{t=1}^{T_1}, \ldots, \{(\bvy_{t,n}, X_{t,n})\}_{t=1}^{T_n} \}$, provided that they are independent and identically distributed (iid). 

The complete-data log-likelihood for one trajectory of the copula HMM is
given by
\begin{align}
    \ell_\text{com}\left(\bveta \mid \bvy\oneT, X\oneT\right) 
    &= \pi_{X_1} + \sum_{t=2}^T \log \gamma_{X_{t-1},X_{t}} \nonumber \\
    &+ \sum_{h=1}^d \log f_{X_{t},h}(y_{t,h}; \lambda_{X_{t},h}) \nonumber \\
    &+ \sum_{t=1}^T \log 
    c_{X_t}\left(F_{X_{t},1}(y_{t,1}; \lambda_{X_{t},1}),\ldots,F_{X_{t},1}(y_{t,d}; \lambda_{X_{t},d}) \mid \theta_{X_t}\right).
    \label{CDLL0}
\end{align}
We define the latent indicators  \begin{align*}
U_{k,t} &= \one{X_t = k}, \quad k \in \sX, \, t=1,\ldots, T\\
V_{j,k,t} &= \one{X_{t-1} = j, X_t = k }, \quad j,k \in \sX, \, t=2,\ldots, T,
\end{align*}
so that the complete data log-likelihood in \eqref{CDLL0} can be rewritten as
\begin{align}
    \ell_\text{com}\left(\bveta \mid \bvy\oneT, X\oneT\right) 
    &= \sum_{k=1}^K U_{k,1} \cdot \log{\pi_{k}} + \sum_{j=1}^K \sum_{k=1}^K \left(\sum_{t=2}^T V_{j,k,t} \right) \log{ \gamma_{j,k}} \nonumber \\ 
    &+ \sum_{k=1}^K \sum_{t=1}^T U_{k,t} \cdot \left( \sum_{h=1}^d \log{f_{k,h}(y_{t,h}; \lambda_{k,h})} \right)  \nonumber \\
    &+ \log{c_k\Big(F_{k,1}(y_{t,1}; \lambda_{k,1}), \ldots, F_{k,d}(y_{t,d}; \lambda_{k,d})\, \Big| \, \theta_k \Big)}. \label{eq:CDLL}
\end{align}
Now consider the E- and M-steps at the $(s+1)$th iteration of the algorithm, assuming that the algorithm is initialized at $\bveta^{(0)}$. The \textbf{E-step} requires the computation of the conditional expectation 
\begin{equation}
    Q(\bveta \mid \bveta^{(s)}) 
    = \mathbb{E}_{\bveta^{(s)}}\left[ \ell_\text{com}\left(\bveta \mid \bvy\oneT, X\oneT \right) \right],
\end{equation}
and to do so requires computing the conditional expectations of all $U_{k,t}$'s and $V_{j,k,t}$'s given the observed data $\bvy_{1:T}$ and current parameter estimates $\bveta^{(s)}$, which are given by
\begin{equation}\label{eq:uhat_def}
    \hat{u}_{j,t}^{(s+1)} = \parPrb{\bveta^{(s)}}{X_t = j \mid \bvy_{1:T}} 
    = \frac{\alpha_{j,t}(\bvy_{1:t}; \bveta^{(s)}) \cdot \beta_{j,t}(\bvy_{(t+1):T}; \bveta^{(s)})}{\sum_{l=1}^K \alpha_{l,t}(\bvy_{1:t}; \bveta^{(s)}) \cdot \beta_{l,t}(\bvy_{(t+1):T}; \bveta^{(s)})}
\end{equation}
and
\begin{equation}\label{eq:vhat_def}
    \hat{v}_{j,k,t}^{(s+1)} 
    = \parPrb{\bveta^{(s)}}{X_{t-1} = j, X_t = k \mid \bvy_{1:T}} 
    = \frac{\alpha_{j,t-1}(\bvy_{1:(t-1)}; \bveta^{(s)}) \cdot {\gamma}_{j,k}^{(s)} \cdot h_k(\bvy_t; \bveta^{(s)}) \cdot \beta_{k,t}(\bvy_{(t+1):T}; \bveta^{(s)})}{\sum_{l=1}^K \alpha_{l,t}(\bvy_{1:t}; \bveta^{(s)}) \cdot \beta_{l,t}(\bvy_{(t+1):T}; \bveta^{(s)})},
\end{equation}
where 
\[
     \alpha_{j,t}(\bvy_{1:t}; \bveta) = \parPrb{\bveta}{\bvY_{1:t} = \bvy_{1:t}, X_t = j}
     \quad \text{and} \quad 
    \beta_{j,t}(\bvy_{(t+1):T}; \bveta) = \parPrb{\bveta}{\bvY_{(t+1):T} = \bvy_{(t+1):T} \mid X_t = j}
\]
are known as the \emph{forward probabilities} and \emph{backward probabilities}, respectively. These ingredients are essential for HMM classification algorithms, and are well-known to admit recursive structures which allow for their efficient computation via dynamic programming. The latter fact is not altered by the addition of the copula density in \eqref{eq:CDLL}.

The \textbf{M-step} requires finding the maximizer $\bveta^{(s+1)}$ of
\begin{align}
    Q(\bveta|\bveta^{(s)}) 
    &= \sum_{k=1}^K \hat{u}_{k,1}^{(s)} \cdot \log{\pi_{k}} \nonumber \\
    &+ \sum_{j=1}^K \sum_{k=1}^K \left(\sum_{t=2}^T \hat{v}_{j,k,t}^{(s)} \right) \cdot \log{ \gamma_{j,k}}  \nonumber \\ 
    &+ \sum_{k=1}^K \sum_{t=1}^T \hat{u}_{k,t}^{(s)} \cdot \left( \log{c_k\Big(F_{k,1}(y_{t,1} \mid \lambda_{k,1}), \ldots, F_{k,d}(y_{t,d} \mid \lambda_{k,d})\, \Big| \, \theta_k \Big)} + \sum_{h=1}^d \log{f_{k,h}(y_{t,h} \mid \lambda_{k,h})} \right). \label{eq:MstepLL}
\end{align}

Since the parameters in the three sums of \eqref{eq:MstepLL} are functionally independent, they can be maximized independently. Using Lagrange multipliers, one can easily show that the maximizer for the initial distribution (as a vector in the standard simplex $\sS^{K-1}$) is
\[
    \bvpi^{(s+1)} 
    = \argmax{\bvpi \in \sS^{K-1}}\left( \sum_{k=1}^K \hat{u}_{k,1}^{(s)} \cdot \log{\pi_k} \right) 
    = (\hat{u}_{1,1}^{(s)}, \ldots, \hat{u}_{K,1}^{(s)}),
\]
while that for the vector of transition probabilities from state $j \in \sX$ is \[
    \bvgamma_{j,\cdot}^{(s+1)} 
    = \argmax{\bvgamma \in \sS^{K-1}} \left(\sum_{k=1}^K \left(\sum_{t=2}^T \hat{v}_{j,k,t}^{(s)} \right) \cdot \log{ \gamma_{k}}\right) 
    = \left(\frac{ \sum_{t=2}^T \hat{v}_{j,1,t}^{(s)}  }{\sum_{k=1}^K \sum_{t=2}^T \hat{v}_{j,k,t}^{(s)}}, \ldots,  \frac{ \sum_{t=2}^T \hat{v}_{j,K,t}^{(s)}  }{\sum_{k=1}^K \sum_{t=2}^T \hat{v}_{j,k,t}^{(s)}}   \right).
\]
The maximizer for the vector of parameters involved in the state-dependent distributions consists of the marginal parameters $\{\lambda_{i,h}\}_{i \in \sX, h \in [d]}$ and copula parameters $\{\theta_i\}_{i \in \sX}$ that jointly maximize 
\begin{equation}
    \sum_{k=1}^K \sum_{t=1}^T \hat{u}_{k,t}^{(s-1)} \cdot \left[ \log c_k \Big(F_{k,1}(y_{t,1}; \lambda_{k,1}), \ldots, F_{k,d}(y_{t,d}; \lambda_{k,d}) \, \Big| \, \theta_k\Big) + \sum_{h=1}^d \log f_{k,h}(y_{t,h} ; \lambda_{k,h}) \right].\label{eq:copulaLL}
\end{equation}
At this point, the M-step falters because the optimization of \eqref{eq:copulaLL} is difficult; when the dimension $d$ is even moderately high, naively applying a numerical optimizer to \eqref{eq:copulaLL} directly is likely to fail, for several reasons. For one, initialization can be challenging; since any copula is grounded, even a \emph{correct} initialization $\lambda_{k,h}^{(0)} = \lambda_{k,h}$ can led us astray, for we can easily have $F_{k,h}(y_{t,h}; \lambda_{k,h}) \approx 0$ when $X_t \neq k$, and hence $c_{k}(\cdots, F_{k,h}(y_{t,h}; \lambda_{k,h}),\cdots) \approx 0$. Unless $\hat{u}_{k,t}^{(0)} \approx 0$ as well --- which would be unusual when we have no \emph{a priori} information about the true underlying states --- evaluation of \eqref{eq:copulaLL} will immediately cause a numerical overflow. We emphasize that this can occur even if all parameters are initialized to their true values.

A second challenge comes from the fact that the log-likelihood surface can have wide regions of flatness; when Hessian-based methods (or quasi-Newton methods) are used for optimization, this can cause instability. For example, we consider a simple version of the model \eqref{eq:state-dependentDF} in which only the copula parameter contains information about the hidden state $\bvY_t \mid (X_{t} = k) \sim C_\mathrm{Frank}\left(\sN(\mu, 1), \sN(\mu, 1) \mid \theta_k\right)$ for $k=1,2$. \autoref{fig:LL_HMMsurfaces} shows three log-likelihood surfaces, anchored by three values of $\mu$. One can see that for moderate values of $\theta_1$ and $\theta_2$, the log-likelihood surface is quite flat.

\begin{figure}[ht]
    \centering
    \includegraphics[width=\textwidth]{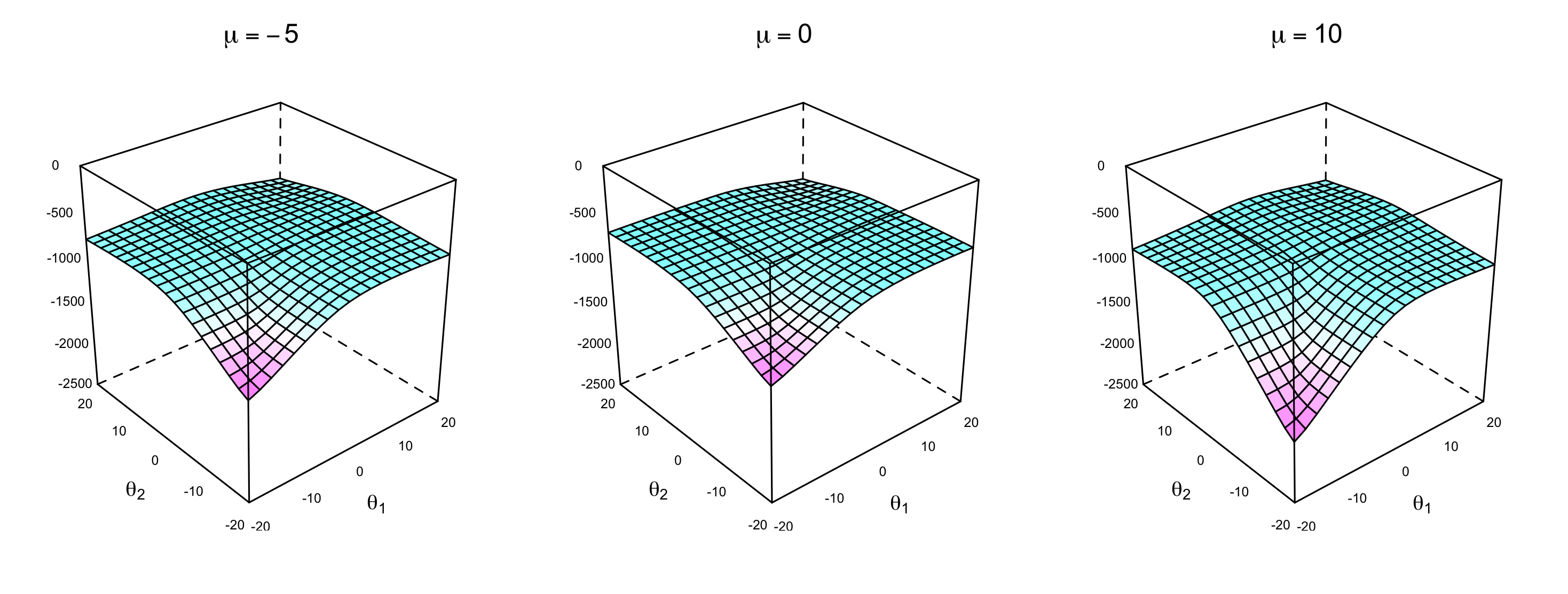}
    \caption{Log-likelihood surfaces based on the HMM $\bvY_t \mid (X_{t} = k) \sim C_\mathrm{Frank}\left(\sN(\mu, 1), \sN(\mu, 1) \mid \theta_k\right)$ for $k=1,2$, with true parameters $\mu = 0$, and $\theta_1 = -\theta_2 =10$ and all initial probabilities and transition probabilities equal to $1/2$.}
    \label{fig:LL_HMMsurfaces}
\end{figure}

\subsection{An Improved Algorithm}\label{sub:improvedlalgo}

In contrast to jointly maximizing all parameters in \eqref{eq:copulaLL}, it is considerably less challenging (if not necessarily easy) to estimate the copula parameters $\{\theta_i\}_{i \in \sX}$ and the marginal parameters $\{\lambda_{i,h}\}_{i \in \sX, h \in [d]}$ \emph{separately}. For instance, likelihood inference for the five classical one-parameter Archimedean families under \emph{known} marginals was studied by \cite{hofert2012likelihood}, who derived concise functional representations for the copula densities and their score functions. Altogether, these allow for maximum likelihood estimation using gradient-based optimization methods. Meanwhile, the difficulty of performing likelihood inference for the marginal distributions themselves depends highly on their specific forms (as for the standard case of HMMs with univariate state-dependent distributions). While closed-form solutions --- or efficient estimation procedures --- for MLEs in the presence of iid data have been long-established for most commonly encountered distributions, the situation is different for the models we introduce here, where the function to be optimized is essentially a weighted sum of log-densities, and closed-form expressions for the maximizing values are rarely encountered. 

If the HMM structure were removed from the model so that the data were assumed to consist of iid observations from $C(F_{1}(\cdot \,; \lambda_{1}), \ldots, F_{d}(\cdot \,; \lambda_d) \big| \theta)$, then \eqref{eq:copulaLL} would reduce to 
\begin{equation}\label{eq:LLonestate}
    \sum_{t=1}^T\left \{  \log{c\Big(F_{1}(y_{t,1} \,; \lambda_1), \ldots, F_{d}(y_{t,d} \,; \lambda_d) \, \Big| \, \theta\Big)} + \log{f_{h}(y_{t,h} ; \lambda_{h})}\right \} .
\end{equation}
Moreover, if the copula were known, an estimate for the marginal parameter $\lambda_h$ could be inferred from the marginal data $y_{1:T,h}$ alone as the standard MLE
\begin{equation}\label{eq:LLonemarginal}
    \tilde{\lambda}_h = \argsup_{\lambda} \, \sum_{t=1}^T \log{f_{h}(y_{t,h}; \lambda)}.
\end{equation}
This is exactly what is done in the first step of the \emph{inference functions for margins (IFM)} method of \cite{joe1996estimation}, which is itself grounded on the theory of inference functions  \citep{mcleish2012theory}. The IFM method starts by finding the $\tilde{\lambda}_h$'s that separately maximize the marginal log-likelihoods \eqref{eq:LLonemarginal} for $h=1, \ldots, d$, and then proceeds by finding the $\tilde{\theta}$ that maximizes the joint log-likelihood \eqref{eq:LLonestate} in which the marginal parameters are set to their previous estimates: \begin{equation}\tilde{\theta} = \argsup_{\theta} \, \sum_{t=1}^T \log{c\Big(F_{1}(y_{t,1} \,; \tilde{\lambda}_1), \ldots, F_{d}(y_{t,d} \,;  \tilde{\lambda}_d) \, \Big| \, \theta\Big)}.\label{eq:LLonecopula}\end{equation}

Section 10.1 of \cite{joe1997multivariate} shows that under standard regularity conditions, $(\tilde{\lambda}_1, \ldots, \tilde{\lambda}_d, \tilde{\theta})$ is a consistent and asymptotically normal estimator of $(\lambda_1, \ldots, \lambda_d, \theta)$, although it is not as efficient as the usual maximum likelihood estimator. 

We propose here a new algorithm, which we call the \emph{expectation-IFM (EIFM) algorithm}, that alleviates the computational difficulties associated with the estimation of $\bveta$. The main idea is to embed the IFM method within the maximization step of the basic EM algorithm. Specifically, while the E-step remains unchanged, in the M-step we replace the single-state log-likelihoods in  \eqref{eq:LLonemarginal} and \eqref{eq:LLonecopula} with weighted averages that are computed over all $K$ states. The log-likelihood for the $k$'th state is weighted by the corresponding state membership probability estimates $\{\hat{u}_{k,t}\}$. Starting with an initial guess $\bveta^{(0)}$, we perform the following steps, iterating over $s \geq 1$ until we reach convergence:

\begin{enumerate}[\bfseries Step 1:]
\item (E-Step) Calculate the conditional expectations:
\begin{enumerate}[\bfseries {Step 1}a:]
    \item Calculate 
    \[
        \hat{u}^{(s+1)}_{j,t} 
        = \frac{\alpha_{j,t}(\bvy_{1:t}; \bveta^{(s)}) \cdot \beta_{j,t}(\bvy_{(t+1):T}; \bveta^{(s)})}{\sum_{l=1}^K \alpha_{l,t}(\bvy_{1:t}; \bveta^{(s)}) \cdot \beta_{l,t}(\bvy_{(t+1):T}; \bveta^{(s)})}
    \]
    for $j \in \sX$ and $t \in \{1, \ldots, T\}$.
    \item Calculate 
    \[
        \hat{v}^{(s+1)}_{j,k,t} 
        = \frac{\alpha_{j,t-1}(\bvy_{1:(t-1)}; \bveta^{(s)}) \cdot {\gamma}_{j,k}^{(s)} \cdot h_k(\bvy_t; \bveta^{(s)}) \cdot \beta_{k,t}(\bvy_{(t+1):T}; \bveta^{(s)})}{\sum_{l=1}^K \alpha_{l,t}(\bvy_{1:t}; \bveta^{(s)}) \cdot \beta_{l,t}(\bvy_{(t+1):T}; \bveta^{(s)})}
    \]
    for $j,k \in \sX$ and $t \in \{2, \ldots, T\}$.
\end{enumerate}

\item (IFM-Step) Estimate model parameters using an IFM approach:
\begin{enumerate}[\bfseries {Step 2}a:]
    \item Estimate the initial distribution and transition probabilities:
    \[
        \bvpi^{(s+1)} = (\hat{u}_{1,1}^{(s)}, \ldots, \hat{u}_{K,1}^{(s)})
    \]
    and
    \[
        \bvgamma_{j,\cdot}^{(s+1)} 
        = \left(\frac{ \sum_{t=2}^T \hat{v}_{j,1,t}^{(s)}  }{\sum_{k=1}^K \sum_{t=2}^T \hat{v}_{j,k,t}^{(s)}}, \ldots,  \frac{ \sum_{t=2}^T \hat{v}_{j,K,t}^{(s)}  }{\sum_{k=1}^K \sum_{t=2}^T \hat{v}_{j,k,t}^{(s)}}   \right)
    \]
    for $j \in \sX$.
    \item Estimate the marginal parameters 
    \[
        \lambda_{k,h}^{(s+1)} 
        = \argsup_{\lambda} \sum_{t=1}^T \hat{u}_{k,t}^{(s+1)} \cdot \log{f_{k,h}(y_{t,h}; \lambda)}
    \]
    for $k \in \{1, \ldots, K\}$ and $h \in \{1, \ldots, d\}$.
    \item Estimate the copula parameters 
    \begin{equation}\label{eq:copulaargsup}
        \tilde{\theta}_k^{(s+1)} 
        = \argsup_{\theta} \sum_{t=1}^T \hat{u}_{k,t}^{(s+1)} \cdot \log{c_k\Big(F_{k,1}(y_{t,1}; \lambda_{k,1}^{(s+1)}), \ldots, F_{k,d}(y_{t,d};  \lambda_{k,d}^{(s+1)}) \, \Big| \, \theta\Big)}
    \end{equation}
    for $k \in \{1, \ldots, K\}$.
\end{enumerate}
\end{enumerate}

Step 2b is  an importance-weighted maximum likelihood estimation; it is typically not more difficult than maximum likelihood estimation in the unweighted case, and the optimal parameter values are sometimes available in closed form. For example, if the marginal distribution is in an exponential family with vector parameter $\boldsymbol{\lambda} \in \Theta \subseteq \R^q$, that is,
\[
    f_{k,h}(y_{t,h}; \boldsymbol{\lambda}) 
    = g(y_{t,h}) \cdot \exp{\bveta(\boldsymbol{\lambda})^\top  \bvT(y_{t,h}) - A(\bveta)},
\]
it is straightforward to show that Step 2b amounts to solving the system
\begin{equation}\label{eq:expfamily}
    \bvJ_{\bveta}(\boldsymbol{\lambda})^\top \nabla A(\boldsymbol{\lambda}) 
    = \frac{1}{\sum_{t=1}^T \hat{u}^{(s)}_{k,t}} \sum_{t=1}^T \hat{u}^{(s)}_{k,t} \bvJ_{\bveta}(\boldsymbol{\lambda})^\top \bvT(y_{t,h})
    = \bvzero
\end{equation}
in $\boldsymbol{\lambda}$. Here $\bvJ_{\bveta}(\boldsymbol{\lambda})$ is the Jacobian of $\bveta(\boldsymbol{\lambda})$, which is simply the identity matrix when the exponential family is in canonical form.

The difficulty of Step 2c depends on the ease of evaluating the copula density $c_k$ and its derivative with respect to $\theta$. For the five ``classical'' one-parameter Archimedean copula families investigated by \cite{hofert2012likelihood}, the functional forms for the densities, generator derivatives, and score functions can be computed using  the \texttt{copula} package \citep{Rcopula} within \texttt{R}. These functions can be used in conjunction with the built-in \texttt{optim} function or an appropriate numerical optimization library to approximate \eqref{eq:copulaargsup} in a straightforward manner.

\section{Theoretical Results}\label{sec:TheoreticalResults}

In this section, we demonstrate that the EIFM algorithm does not belong to the class of generalized EM algorithms \citep{dempster1977maximum}, and therefore a theoretical analysis is needed to justify its use. We start off with a purely algorithmic perspective in \autoref{sub:algorithm}, describing the circumstances under which the sequence of estimates produced by the EIFM algorithm will converge. In \autoref{sub:estimator}, we show that under relatively weak assumptions, the estimator produced by the algorithm upon convergence will have desirable statistical properties.

\subsection{Analysis of the Algorithm}\label{sub:algorithm}

In order to justify the convergence of the EIFM algorithm, we temporarily set aside its statistical content and view it as a deterministic nonlinear optimization procedure. The algorithm is not an EM algorithm in the traditional sense, because the expectation of the complete-data log-likelihood is not maximized at each iteration. This is mainly a consequence of the fact that
\begin{equation}
    \sum_{t=1}^T \hat{u}_{t} \cdot \log{f_{h}(y_{t,h}; \lambda_{h}^{(s)})} 
    \leq \sum_{t=1}^T \hat{u}_{t} \cdot \log{f_{h}(y_{t,h}; \lambda_{h}^{(s+1)})}, \qquad h \in \{1, \ldots, d\} 
\label{eq:MLEcounterexample}\end{equation} 
does not imply 
\[
    \sum_{t=1}^T \hat{u}_{t} \cdot \log{c\Big(F_{1}(y_{t,1}; \lambda^{(s)}_{1}), \ldots, F_{d}(y_{t,d}; \lambda^{(s)}_{d})\, \Big| \, \theta^{(s)} \Big)} 
    \leq \sum_{t=1}^T \hat{u}_{t} \cdot \log{c\Big(F_{1}(y_{t,1}; \lambda^{(s+1)}_{1}), \ldots, F_{d}(y_{t,d}; \lambda^{(s+1)}_{d})\, \Big| \, \theta^{(s)} \Big)}.
\]
For a simple counterexample, one can take a single observation (i.e., $T=1$) with $Y_1 \sim \mathrm{Exp}(\lambda_1)$ and $Y_2 \sim \mathrm{Exp}(\lambda_2)$ such that $(F_1(Y_1),F_2(Y_2))$ has a Farlie–Gumbel–Morgenstern copula with density $c(u,v \mid \theta) = 1 + \theta(2u-1)(2v-1)$. Suppose the algorithm were initialized at $\lambda_1^{(0)} = \lambda_2^{(0)} = 1$ and any $\theta^{(0)} \in (0,1)$. Step 2b of the algorithm produces the standard univariate MLEs $\lambda_1^{(1)} = y_1^{-1}$ and $\lambda_2^{(1)} = y_2^{-1}$, which satisfy \eqref{eq:MLEcounterexample} by construction; however, the inequality
\[
    \log{c\left(F_1(y_1; \lambda_1^{(0)}), F(y_2;\lambda_2^{(0)}) \mid \theta^{(0)}\right)} 
    \leq  \log{c\left(F_1(y_1; \lambda_1^{(1)}), F(y_2;\lambda_2^{(1)}) \mid \theta^{(0)}\right)}
\]
is equivalent to 
\[
    (2e^{-y_1} - 1)(2e^{-y_2} - 1) \leq (2e^{-1} - 1)^2,
\]
which immediately fails when $y_1,y_2 > 1$. 

The above counterexample shows that despite sharing a similar strategy with the ECM algorithm \citep{meng1993maximum}, the EIFM algorithm cannot be justified in the same way as any \emph{generalized EM algorithm} \citep{dempster1977maximum}, as these all rely on increasing the objective function at each iteration. Clearly, one needs to identify the conditions under which the EIFM algorithm will converge, since its  sequential updating rule does not offer convergence guarantees  without stronger assumptions. Our analysis here follows roughly the one given for the \emph{ES algorithm} by \cite{elashoff2004algorithm}, as the aim of that procedure is also to iteratively solve unbiased estimating equations in the presence of missing data. However, the dependence inherent in the HMM data precludes a direct application of the ES algorithm.

We regard the $\hat{u}_{j,t}$'s and $\hat{v}_{j,k,t}$'s as parameters themselves (in a non-statistical context), and collect them and the parameter of interest $\bveta$ into a larger parameter vector $\bvxi = (\hat{\bvu}, \hat{\bvv}, \bveta)$ of length $N := TK + (T-1)K^2 + K + K^2 + dK + K$. For accounting purposes, let
\begin{align*}
    N_1 &= 0\\
    N_2 &= N_1 + TK\\
    N_3 &= N_2 + (T-1)K^2\\
    N_4 &= N_3 + K\\
    N_5 &= N_4 + K^2\\
    N_6 &= N_5 + dK.
\end{align*}
We construct a vector-valued function $\bvg(\bvxi) = \left(g_1(\bvxi), \ldots, g_N(\bvxi)\right)^\top$ with each $g_i:\R^N \to \R$ defined as follows: 
\begin{align*}
    g_{N_1 + jt}(\bvxi) &= \hat{u}_{j,t} 
        - \frac{\alpha_{j,t}(\bvy_{1:t}; \bveta) \cdot \beta_{j,t}(\bvy_{(t+1):T}; \bveta)}{\sum_{l=1}^K \alpha_{l,t}(\bvy_{1:t}; \bveta) \cdot \beta_{l,t}(\bvy_{(t+1):T}; \bveta)}, &&j \in \sX, \, \, t \in \{1, \ldots, T\} \\ 
    g_{N_2 + jkt}(\bvxi) &= \hat{v}_{j,k,t} 
        - \frac{\alpha_{j,t-1}(\bvy_{1:(t-1)}; \bveta) \cdot {\gamma}_{j,k} \cdot h_k(\bvy_t; \bveta) \cdot \beta_{k,t}(\bvy_{(t+1):T}; \bveta)}{\sum_{l=1}^K \alpha_{l,t}(\bvy_{1:t}; \bveta) \cdot \beta_{l,t}(\bvy_{(t+1):T}; \bveta)}, &&  j,k \in \sX, \, \, t \in \{2, \ldots, T\}\\
    g_{N_3 + j}(\bvxi) &= \pi_j - \hat{u}_{j,1}, && j \in \sX \\
    g_{N_4 + jk}(\bvxi) &= \gamma_{j,k} \cdot \sum_{l=1}^K\sumtT \hat{v}_{j,l,t} - \sumtT \hat{v}_{j,k,t}, && j,k \in \sX \\
    g_{N_5 + jh}(\bvxi) &= \sumtT \hat{u}_{j,t} \cdot \frac{\partial}{\partial \lambda} \log{f_{j,h}(y_{t,h}; \lambda)}, && j \in \sX, \, h \in \{1, \ldots, d\} \\
    g_{N_6 + j}(\bvxi) &= \sumtT \hat{u}_{j,t} \cdot \frac{\partial}{\partial \theta} \log{c_j\left(F_{j,1}(y_{t,1}; \lambda_{j,1}), \ldots, F_{j,d}(y_{t,d}; \lambda_{j,d}) \mid \theta\right)}, && j \in \sX .
\end{align*}
From the description of the EIFM algorithm in \autoref{sub:improvedlalgo}, it can be seen that the entire $(s+1)$'th iteration of the algorithm exactly corresponds to updating $\bvxi^{(s)} \to \bvxi^{(s+1)}$ by setting $\xi_i^{(s+1)}$ as the solution to the univariate problem 
\begin{equation}\label{eq:GSupdate}
    g_i(\xi^{(s+1)}_1, \ldots, \xi_{i-1}^{(s+1)}, \xi, \xi^{(s)}_{i+1},\ldots, \xi^{(s)}_N) = 0,
\end{equation}
for each $i=1,\ldots,N$ (in practice, many of these sub-updates are performed in parallel --- for example, the sub-updates $N_4 + 1$ to $N_5$ altogether correspond to Step 2a of the algorithm). Supposing the sequence $\{\bvxi^{(s)}\}_{s \geq 1}$ converges to some $\bvxi^* \in \R^N$, the limiting vector will satisfy $\bvg(\bvxi^*) = \bvzero$, and the sub-vector $\bveta^*$ will be taken as our estimator of $\bveta$, whose statistical properties are studied in \autoref{sub:estimator}. The EIFM algorithm is thus an example of a \emph{nonlinear Gauss-Seidel method}, or more generally, a \emph{nonlinear successive over-relaxation (SOR) method} \citep{ortega2000iterative}. The local convergence of such algorithms depend on the behaviour of the Jacobian $\bvJ_{\bvg}(\bvxi)$ of $\bvg$ in a neighborhood of some solution to $\bvg(\bvxi) = \bvzero$ such that a certain transformation $\tilde{\bvg}:\R^N \to \R^N$ of $\bvg$ is (locally) required to be a contraction mapping with $\bvxi^*$ as a fixed point. The critical theorem is the following:

\begin{theorem}[\cite{ortega2000iterative}, Theorem 10.3.5]\label{thm:SORconvergence}
Let $\bvg:D \subset \R^N \to \R^N$ be continuously differentiable in an open neighborhood $S_0$ of $\bvxi^*$ such that $\bvg(\bvxi^*) = \bvzero$. Decompose the Jacobian into diagonal, strictly lower, and strictly upper matrices as $\bvJ_{\bvg}(\bvxi) = \bvD(\bvxi) - \bvL(\bvxi) - \bvU(\bvxi)$, and suppose that $\bvD(\bvxi^*)$ is nonsingular. If $\rho([\bvD(\bvxi^*) - \bvL(\bvxi^*)]^{-1} \bvU(\bvxi^*)) < 1$, then there exists an open ball $B(\bvxi^*, \delta) \subset S_0$ such that for any $\bvxi^{(0)} \in B(\bvxi^*, \delta)$, there exists a unique sequence $\{\bvxi^{(s)}\} \subset B(\bvxi^*, \delta)$ satisfying the nonlinear SOR prescription, such that $\lim_{s \to \infty} \bvxi^{(s)} = \bvxi^*$.
\end{theorem}

See \cite{ortega2000iterative} for a proof. Note here that for a general matrix $\bvA \in \C^{N \times N}$ with eigenvalues $\nu_1, \ldots, \nu_N$, the function $\rho(\bvA)$ refers to the \emph{spectral radius} of $\bvA$, defined as $\max\{|\nu_1|, \ldots, |\nu_N|\}$. This theorem implies that the EIFM algorithm, once sufficiently close to a local solution $\bvxi^*$ of $\bvg(\bvxi) = \bvzero$, will produce a sequence converging superlinearly to $\bvxi^*$, provided that the spectral radius of the matrix $[\bvD(\bvxi^{(s)}) - \bvL(\bvxi^{(s)})]^{-1} \bvU(\bvxi^{(s)}))$ is less than 1. Verifying the nonsingularity of $\bvD(\bvxi^{(s)})$ is straightforward; it is easy to see that the first $N_5$ diagonal entries are equal to 1, while the remaining entries are of the form
\[
    \left.\sumtT \hat{u}^{(s)}_{j,t} \cdot \frac{\partial^2}{\partial \lambda^2} \log{f_{j,h}(y_{t,h}; \lambda)} \right|_{\lambda = \lambda^{(s)}_{j,h}} 
    \quad \text{or} \quad
    \left.\sumtT \hat{u}^{(s)}_{j,t} \cdot \frac{\partial^2}{\partial \theta^2} \log{c_j\left(F_{j,1}(y_{t,1}; \lambda^{(s)}_{j,1}), \ldots, F_{j,d}(y_{t,d}; \lambda^{(s)}_{j,d}) \mid \theta\right)} \right|_{\theta = \theta^{(s)}_j},
\]
which can be calculated directly. Some of these calculations may be unnecessary with prior knowledge of the involved densities; for example, the term on the left is certain to be negative if the mapping $\lambda \mapsto f_{j,h}(y; \lambda)$ is known to be strictly concave. The remaining elements of $\bvJ_{\bvg}$ at $\bvxi^{(s)}$ and the corresponding spectral radius $\rho([\bvD(\bvxi^{(s)}) - \bvL(\bvxi^{(s)})]^{-1} \bvU(\bvxi^{(s)}))$ can, in principle, be computed exactly following the $s$'th iteration of the EIFM algorithm, although in practice approximating these quantities numerically would be considerably easier than explicitly calculating all of the required derivatives.

\subsection{Analysis of the Estimator}\label{sub:estimator}

We now discuss conditions under which the subvector $\bveta^*$ of the solution $\bvxi^*$ ultimately produced by the EIFM algorithm --- now regarded as a statistical estimator $\bveta^*(\bvY\oneT)$ --- is consistent for the true parameter $\bveta_0$ for the model described in \autoref{sub:genmod} which is assumed to have generated the data. While our analysis of the algorithm's convergence  in \autoref{sub:algorithm} uses the same underlying technique as \cite{elashoff2004algorithm} for their ES algorithm, we cannot emulate their proof of consistency. Unlike their data, which are assumed to be iid, thus allowing them to use standard techniques (see \cite{godambe1960optimum}, for example) to prove the consistency and asymptotic normality of their estimator,  the latent variables in our setup --- namely, the conditional state membership indicators $U_{k,t}$'s and the $V_{j,k,t}$'s --- are  not independent. Their dependency is due to their strong connection to the underlying Markov structure of the HMM. We instead follow the approach of \cite{jensen2011asymptotic}, who studied the asymptotic properties of M-estimators used for HMMs. 

Specifically, we assume that the EIFM algorithm has produced a vector $\bvxi^* = (\bvu^*, \bvv^*, \bveta^*)$ which satisfies $\bvg(\bvxi^*) = \bvzero$, where $\bvg$ is as defined in \autoref{sub:algorithm}. To show that this estimator has desirable asymptotic properties under mild regularity conditions, it suffices to show that the lower $p = K + K^2 + dK + K$ components of $\bvg$ --- those corresponding to the FM-step of the algorithm --- constitute a system of unbiased estimating equations for $\bveta$ conditional on $\bvY\oneT$ after the latent variables $X\oneT$ have been marginalized out. Specifically, we construct a new vector-valued function $\bvpsi_T(\bveta; \bvY\oneT)$ as follows. First, we define the random $\R$-valued functions
\begin{align}
    \phi_{N_3 + j,t}\left(\bveta; X_{(t-1):t}, \bvY_t\right) &= \begin{cases} \pi_j - U_{j,1}, & t = 1\\ 0, & t > 1\end{cases}, && j \in \sX \label{eq:UE1} \\
    \phi_{N_4 + jk}\left(\bveta; X_{(t-1):t}, \bvY_t\right) &= \gamma_{j,k} \cdot \sum_{l=1}^K V_{j,l,t} - V_{j,k,t}, && j,k \in \sX \label{eq:UE2} \\
    \phi_{N_5 + jh}\left(\bveta; X_{(t-1):t}, \bvY_t\right) &= U_{j,t} \cdot \frac{\partial}{\partial \lambda} \log{f_{j,h}(Y_{t,h}; \lambda)}, &&  j \in \sX, \,  h \in \{1, \ldots, d\} \label{eq:UE3} \\
    \phi_{N_6 + j}\left(\bveta; X_{(t-1):t}, \bvY_t\right) &=  U_{j,t} \cdot \frac{\partial}{\partial \theta} \log{c_j\left(F_{j,1}(Y_{t,1}; \lambda_{j,1}), \ldots, F_{j,d}(Y_{t,d}; \lambda_{j,d}) \mid \theta\right)}, && j \in \sX. \label{eq:UE4}
\end{align}
Note that as functions, only \eqref{eq:UE1} depends on $t$; this is required in order to capture the initial distribution parameters $\bvpi$. The remaining functions are free of $t$ because of the time-homogeneity of the underlying Markov chain. Now, let
\[
    \bvphi_t\left(\bveta; X_{(t-1):t}, \bvY_t\right) = \left(\phi_{N_3 + 1,t}\left(\bveta; X_{(t-1):t}, \bvY_t\right),\ldots, \phi_{N_3 + p}\left(\bveta; X_{(t-1):t}, \bvY_t\right) \right)^\top
\]
and
\[
    \bvpsi_T(\bveta; \bvY\oneT) = \sumtT \mathbb{E}_{\bveta}\left[\bvphi_t\left(\bveta; X_{(t-1):t}, \bvY_t\right) \mid \bvY\oneT \right].
\]
From \eqref{eq:uhat_def} and \eqref{eq:vhat_def} and the fact that the algorithm has converged to the fixed point $\bvxi^*$, it is clear that 
\[
    u^*_{j,t} = \parPrb{\bveta^*}{X_t = j \mid \bvY\oneT} 
    \quad \text{and} \quad
    v^*_{j,k,t} = \parPrb{\bveta^*}{X_{t-1} = j, X_t = k \mid \bvY\oneT},
\]
and from these it is easy to verify that $\bveta^*$ solves $\bvpsi_T(\bveta; \bvY\oneT) = \bvzero$. Its use as an estimator of $\bveta_0$ is justified by the following fundamental fact:

\begin{theorem}
Under standard regularity conditions, $\bvpsi_T(\bveta; \bvY\oneT)$ defines an unbiased estimating equation. 
\end{theorem}
\begin{proof}
It suffices to check that each of \eqref{eq:UE1}--\eqref{eq:UE4} have expectations (with respect to the joint distribution of $\bvY\oneT$ and $X\oneT$) of $0$. Below, all expectations are taken with respect to a generic $\bveta \in \Theta$; we suppress the subscript for ease of notation. To begin with, 
\[
    \mathbb{E}\left[ \mathbb{E} \left[ \phi_{N_3 + j, 1}\left(\bveta; X_{(t-1):t}, \bvY_t\right) \mid \bvY\oneT \right] \right]
    = \pi_j - \mathbb{P}\left(X_1 = k \right)
    = 0
\]
and, similarly,
\begin{align*}
    \mathbb{E}\left[ \mathbb{E} \left[ \phi_{N_4 + jk}\left(\bveta; X_{(t-1):t}, \bvY_t\right) \mid \bvY\oneT  \right] \right]
    &= \gamma_{j,k} \cdot \sum_{l=1}^K \mathbb{P}\left(X_{t-1} = j, X_t = l\right) - \mathbb{P}\left(X_{t-1} = j, X_t = k\right)\\
    &= \gamma_{j,k} \cdot \mathbb{P}\left(X_{t-1} = j\right) - \mathbb{P}\left(X_t = k \mid X_{t-1} = j\right) \cdot \mathbb{P}\left( X_{t-1} = j\right)\\
    &= 0.
\end{align*}
Moreover, under standard regularity conditions required for the Bartlett identity
\begin{equation}\label{eq:IFM1}
    \mathbb{E}\left[\frac{\partial}{\partial \lambda} \log{f_{k,h}(Y_{t,h}; \lambda)}\right] = 0
\end{equation} to hold, we have 
\[
    \mathbb{E}\left[ \mathbb{E} \left[ \phi_{N_5 + jh}\left(\bveta; X_{(t-1):t}, \bvY_t\right) \mid X\oneT  \right] \right]
    = \mathbb{E}\left[ U_{j,t} \cdot \mathbb{E} \left[ \frac{\partial}{\partial \lambda} \log{f_{j,h}(Y_{t,h}; \lambda)} \right] \right]
    = 0.
\]
Likewise, assuming the same sorts of regularity conditions that permit the interchange of integration and differentiation, we have
\begin{align}\label{eq:IFM2}
\begin{split}
    \mathbb{E} \left[ \phi_{N_6 + j}\left(\bveta; X_{(t-1):t}, \bvY_t\right)\right]
    &=\mathbb{E} \left[ \frac{\partial}{\partial \theta} \log{c_j\left(\bvF_j(\bvY_t; \bvlambda_j) \mid \theta\right)} \right]\\
    &= \int_{\R^d} \frac{\partial}{\partial \theta} c_j\left(\bvF_j(\bvy_t; \bvlambda_j) \mid \theta\right) \prod_{h=1}^d f_{j,h}(y_{t,h}; \lambda_{j,h})\dif \bvy_t\\
    &= \int_{[0,1]^d} \frac{\partial}{\partial \theta} c_j\left(\bvu \mid \theta \right) \dif \bvu\\
    &= \frac{\dif}{\dif \theta} \int_{[0,1]^d} c_j\left(\bvu \mid \theta \right) \dif \bvu\\
    &= 0.
\end{split}
\end{align}
where $\bvF_j(\bvy_t; \bvlambda_j) := \left(F_{j,1}(y_{t,1}; \lambda_{j,1}), \ldots, F_{j,d}(y_{t,d}; \lambda_{j,d}) \right)$ and we have made the usual change of variables $u_h = F_{j,h}(y_{t,h}; \lambda_{j,h})$ for $h \in \{1, \ldots, d\}$. This gives  
\[
    \mathbb{E}\left[ \mathbb{E} \left[ \phi_{N_6 + j}\left(\bveta; X_{(t-1):t}, \bvY_t\right) \mid X\oneT  \right] \right] 
    = \mathbb{E}\left[U_{j,t} \cdot \mathbb{E} \left[ \frac{\partial}{\partial \theta} \log{c_j\left( \bvF_j(\bvy_t; \bvlambda_j) \mid \theta\right)} \right] \right]
    = 0,
\]
which finishes the proof. 
\end{proof}

The IFM method is justified by \eqref{eq:IFM1} and \eqref{eq:IFM2}  \citep[see Subsection 4.1 of][] {joe1996estimation}. The  solution $(\tilde{\lambda}_1, \ldots, \tilde{\lambda}_d, \tilde{\theta})$ of these equations --- at least in the presence of iid data --- is shown to be a consistent and asymptotically normal estimator of the true data-generating parameter as $T \to \infty$. The asymptotic variance of the estimators is given by the inverse Godambe information matrix. However, in order to account for the underlying dependence structure of the HMM, stronger regularity conditions are necessary for the same conclusion to hold. These conditions are encapsulated in the findings of \cite{jensen2011asymptotic}, which establishes a central limit theorem for the sequence of estimators produced by the EIFM algorithm. The relevant assumptions are restated in \autoref{app:regconditions} using our notation; we present the main theorem here:

\begin{theorem}[\cite{jensen2011asymptotic}, Theorem 1]\label{thm:jensensequence}
Let $\bvG_t = \mathrm{Var}_{\bveta_0} \hspace{-0.25em}\left( \bvpsi_t(\bveta_0; \bvY_{1:t}) \right)$ and $\bvH_t = \mathbb{E}_{\bveta_0}\left[ - \bvJ_{\bvpsi_t(\cdot; \bvY_{1:t})}(\bveta_0)\right]$. Under Assumptions 1-4 of \autoref{app:regconditions}, there exists a consistent sequence $\{\hat{\bveta}_T\}$ solving $\bvpsi_T(\bveta; \bvY\oneT) = \bvzero$ such that \[\sqrt{T} \cdot \bvG_T^{-1/2} \bvH_T(\hat{\bveta}_T - \bveta_0) \inD \sN(\bvzero, I)\]as $T \to \infty$ under $\bveta_0$.
\end{theorem}

See \cite{jensen2011asymptotic} for a proof. 

In all but the most trivial situations, the  variances and expectations needed to calculate  $\bvG_T^{-1}$ and, respectively, $\bvH_T$ are analytically intractable. Therefore, we suggest two alternative approaches for estimating standard errors and confidence intervals for $\bveta^*$. In both cases, we start by sampling iid copies $(\bvY\oneT^{(1)}, X\oneT^{(1)}), \ldots, (\bvY\oneT^{(n)}, X\oneT^{(n)})$ from the model under $\bveta^*$. The first approach involves estimating $\bvG_T$ and $\bvH_T$ via Monte Carlo, using $\bveta^*$ as a plug-in estimator for $\bveta_0$. Using the generated samples, we compute
\[
    \hat{\bvG}_T = \frac{1}{n-1}\sumin \left(\bvpsi_T(\bveta^*; \bvY\oneT^{(i)}) -  \bar{\bvpsi}_T(\bveta^*; \bvY\oneT)\right) \left(\bvpsi_T(\bveta^*; \bvY\oneT^{(i)}) -  \bar{\bvpsi}_T(\bveta^*; \bvY\oneT)\right)^\top
\]
where $\bar{\bvpsi}_T(\bveta^*; \bvY\oneT) = \frac{1}{n}\sumin \bvpsi_T(\bveta^*; \bvY\oneT^{(i)})$, and 
\[
    \hat{\bvH}_T = \frac{1}{n} \sumin \left[-\bvJ_{\bvpsi_T(\cdot; \bvY\oneT^{(i)})}(\bveta^*)\right].
\]
Note that the entries of each $\bvpsi_T(\bveta^*; \bvY\oneT^{(i)})$ can be computed by replacing the $U_{j,t}$'s and $V_{j,k,t}$'s in \eqref{eq:UE1}--\eqref{eq:UE4} with their conditional expectations given $\bvY\oneT$ (explicit forms for these are provided by \autoref{lemma:CP1}).  The Jacobian that is required to compute $\hat{\bvH}_T$ can be challenging to derive using symbolic computation software. 

The second approach relies on the  parametric bootstrap estimator of \cite{efron1994introduction}, as  described by \cite{zucchini2017hidden} for HMMs.  Specifically, for each $ i \in \{1, \ldots, n\}$, we re-run the EIFM algorithm on $(\bvY\oneT^{(i)}, X\oneT^{(i)})$ to produce a new estimator $\hat{\bveta}^{(i)}$, and then estimate the covariance matrix of $\bveta^*$ via 
\[
    \hat{\bvV}_{\bveta^*} = \frac{1}{n-1} \sumin \left(\hat{\bveta}^{(i)} - \bar{\bveta} \right)\left(\hat{\bveta}^{(i)} - \bar{\bveta} \right)^\top,
\]
where $\bar{\bveta} = \frac{1}{n} \sumin \hat{\bveta}^{(i)}$. Computation time is reasonable for both methods, because the required tasks are parallelizable. Once all the first and second order derivatives are computed, however, the Monte Carlo method is faster to implement (since it does not require re-fitting the model), but may be prone to numerical errors when the true parameter $\bveta_0$ lies near the boundary of $\bvTheta$.

\section{Simulations and Applications}\label{sec:simsandapps}

\subsection{A Simulation Study}\label{sub:simulationstudies}

The central aim of our simulation study was to demonstrate that introducing the copula into the HMM formulation is advantageous when the dependence structure between observed variables is informative about the latent process, and that it is not harmful otherwise. Secondly, it provides a proof of concept for the EIFM algorithm. Finally, we demonstrate the model's performance under copula misspecification.

All generative models have 3-state HMMs with the same transition probability matrix 
\[
    \bvGamma = \begin{pmatrix}
    0.5 & 0.25 & 0.25 \\
    0.25 & 0.5 & 0.25 \\
    0.25 & 0.25 & 0.5 \\
    \end{pmatrix}.
\]
and initial distribution $\bvpi = (0, 1, 0)$. The marginal distributions are Gaussian and the dependence structure is provided by a Frank copula, whose copula parameter varies between models. Specifically, for $i \in \{1, 2, 3, 4\}$, the $i$'th model has the stochastic representation
\[
    \bvY_t^{(i)} \mid (X_t^{(i)} = k) \sim C_\text{Frank}\left(\sN(\mu_{k,1}, \sigma = 0.5), \sN(\mu_{k,2}, \sigma = 0.5) \mid \theta_{i,k}\right), \quad k \in \{1, 2, 3\}
\]
where $\mu_{k,h} = k + 3 \cdot \one{h=2}$ and 
\begin{align*}
    \bvtheta_1 &= (30, 30, 30)\\
    \bvtheta_2 &= (5, 30, 30)\\
    \bvtheta_3 &= (5, 5, 30)\\
    \bvtheta_4 &= (5, 5, 5).
\end{align*}

We refer to each $i$ as a dependence ``scenario''; thus, the first scenario features strong dependence in all states ($\theta = 30$, or $\tau \approx 0.874$), while the second features weaker dependence in the first state ($\theta = 5$ or $\tau \approx 0.457$), and so on. Note that in the first and fourth scenarios, the copula parameters within each state-dependent distribution are identical, and therefore only the marginals are informative of the underlying states; this contrasts with the numerical example of \autoref{sub:numericalillustration}, in which \emph{only} the copula informed the hidden states. For the $i$th scenario, we generated a time series of length $T = 500$ from the corresponding model with true parameter $\bveta = (\mu_{1,1},\ldots, \mu_{3,2}, \bvtheta_i)$.

First, we assessed the EIFM algorithm's ability to perform both parameter estimation and state classification. Note that optimal values for the $\mu_{k,h}$'s are available in closed form at each iteration of the algorithm via \eqref{eq:expfamily}, thereby limiting any potential numerical obstacles in the algorithm (initializing values aside) to the estimation of the $\theta_{i,k}$'s in Step 2c. To provide the EIFM algorithm with plausible initial parameter values, $\bveta_i^{(0)}$, for each scenario $i$, we first ran an initial EM algorithm assuming independent Gaussian state-dependent distributions with $\mu^{(0)}_{k,h} = (-1)^{h} \cdot |Z_{k,h}|$ and $Z_{k,h} \iid \sN(0,1)$. Using the parameter estimates from this initial run, we used local decoding to make the classifications $\hat{X}^{(i)}_1, \ldots, \hat{X}^{(i)}_T$ and then set 
\[
    \hat{\tau}_{i,k}^{(0)} =
    \frac{2}{|C_{i,k}|\cdot(|C_{i,k}|-1)}\sum_{\substack{s,t \, \in \,  C_{i,k} \\ s < t}}  \mathrm{sgn}\left(y^{(i)}_{1,s} - y^{(i)}_{1,t} \right) \cdot \mathrm{sgn}\left(y^{(i)}_{2,s} - y^{(i)}_{2,t} \right),
\]
where $C_{i,k} = \{t: \hat{X}^{(i)}_t = k\}$; we then set $\theta_{i,k}^{(0)} = \theta_\text{Frank}(\hat{\tau}_{i,k}^{(0)})$. In words, our initial estimate for $\theta_{i,k}$ was based on the empirical Kendall's $\tau$ calculated on the subset of observations which was predicted (by the independent Gaussians model) to have arisen from state $k$. The remaining parameters, common to both the independent Gaussians model and the true data-generating model, were initialized in the second stage algorithm as their corresponding estimates from the initial run.

The entire process was independently replicated 20 times in parallel; we show boxplots of the resulting parameter estimates in \autoref{fig:simulationPlots}. The estimates are observed to be unbiased, and the copula parameters are generally estimated well. \autoref{fig:classificationPlots} shows the overall accuracy of the state classifications (via local decoding) across the four dependence scenarios; label-switching was addressed by selecting the best permutation of the predicted labels for each scenario and each simulation, and then averaging the results. \autoref{fig:classificationPlots} also shows classification accuracy for the state predictions produced by two misspecified models: one in which the Frank copula is replaced by the independence copula (which corresponds to the independent Gaussians model used to generate initial values for the EIFM algorithm), and another in which it is replaced by a Gauss copula. It is clear that across virtually all states and scenarios, the accuracy of the true model was equal or better than that of both misspecified models.

\begin{figure}[H]
    \centering
    \includegraphics[width=\textwidth]{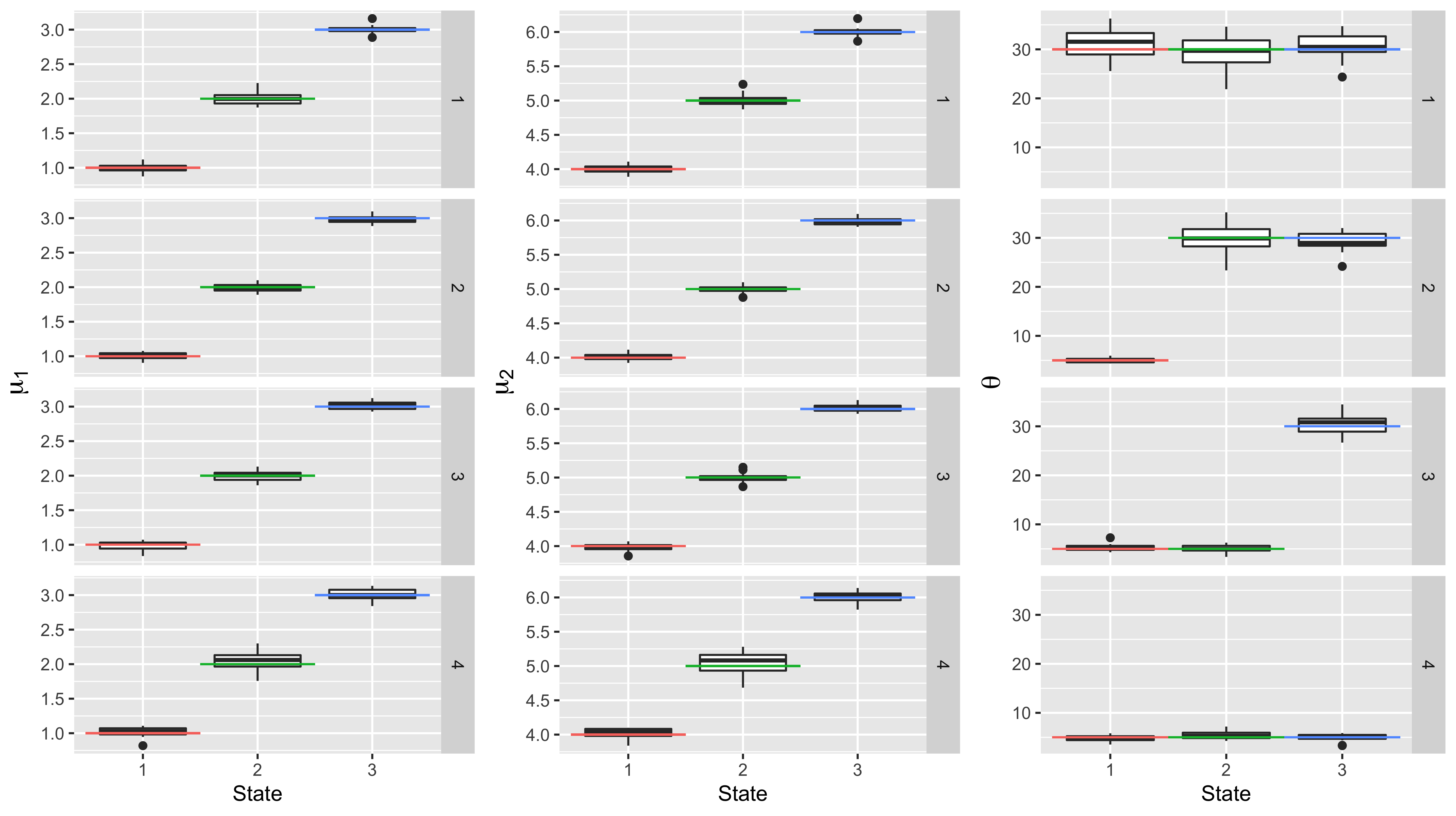}
    \caption{Parameter estimates based on 20 independent simulations and EIFM algorithm runs for the bivariate 3-state HMMs described in \autoref{sub:simulationstudies}. The $i$'th row corresponds to the $i$'th ``scenario''. Colored horizontal lines indicate true parameter values.}
    \label{fig:simulationPlots}
\end{figure}

\begin{figure}[H]
    \centering
    \includegraphics[width=\textwidth]{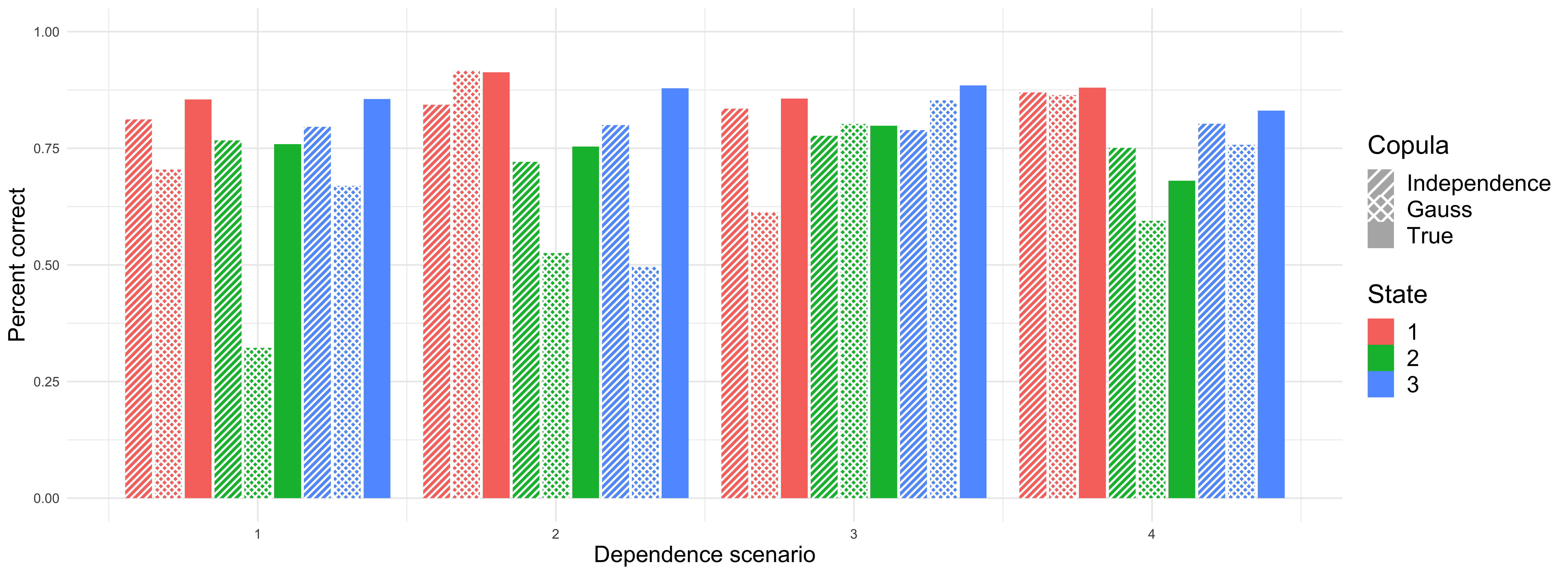}
    \caption{Average percent of correct state classifications for the bivariate 3-state HMMs under the true model (solid bars) and the misspecified models featuring the independence copula (striped bars) and the Gauss copula (hatched bars) described in \autoref{sub:simulationstudies}.}
    \label{fig:classificationPlots}
\end{figure}

\subsection{Occupancy Detection}\label{sub:occupancy}

The ability to detect whether a room is occupied using sensor data (such as temperature and $CO_2$ levels) can potentially reduce unnecessary energy consumption by automatically controlling HVAC and lighting systems, without the need for motion detectors or other methods that could constitute an invasion of privacy \citep{candanedo2016accurate}. To assess the performance of the copula-within-HMM model and the EIFM algorithm on real-world data, we consider three publicly-available labelled datasets presented by \cite{candanedo2016accurate} consisting of multivariate time series of four environmental measurements (light, temperature, humidity, $CO_2$) and one derived metric (the humidity ratio $W$, whose construction is described in \cite{candanedo2016accurate}) captured in an office room in Belgium, along with binary indicators for whether the room was occupied at the time of measurement. 

In \cite{candanedo2017methodology} various HMMs were fit to this data via the EM algorithm using the \texttt{depmixS4} package \citep{depmixS4_2010}; the models featured either one or two (lagged) sensor measurements as the observed data, and the occupancy indicators as the hidden states. The univariate data was taken to be normally distributed, and the bivariate data was assumed to feature independent normal margins. The performance of the models was assessed by comparing predictions made by local decoding to the ground-truth labels via the zero-one loss (as well as several other measures commonly used in binary classification, such as sensitivity and specificity).

We also illustrate our model using bivariate data. We processed the data in the same manner as \cite{candanedo2017methodology}: we averaged the data across five-minute periods, and then calculated the lagged differences to obtain time series of lengths 1629, 533, and 1951 respectively. We selected as our observed data the series $\bvY_t = ((CO_2)_t, W_t)$ and chose the marginal distributions and copulas for each state using several heuristics. First, we selected several common families of parametric copulas (the Frank, Clayton, Gumbel, Joe, and Gauss families), and for each we carried out a goodness-of-fit test based on the pseudo-observations using the multiplier bootstrap method \citep{kojadinovic2011fast}; we then chose the parametric family based on the lowest corresponding Cram\'er-von Mises test statistic. This process yielded a Clayton copula for State 1 and a Frank copula for State 2 (see \autoref{tab:GoFcopulas}), choices which were corroborated by scatterplots of the pseudo-observations (see \autoref{fig:pobs_occupancy}). We used a similar criterion to select the marginal distributions, limiting our choice of distributions to those in exponential families (with appropriate supports); the Gaussian distribution gave the best results for both marginal components within both states.

Labelling the unoccupied state as `1' and the occupied state as `2', our model thus has the stochastic representation
\begin{align*}
    \bvY_t \mid (X_t = 1) &\sim C_\text{Clayton}\left(\sN(\mu_{1,1}, \sigma^2_{1,1}), \sN(\mu_{1,2}, \sigma^2_{1,2}) \mid \theta_{1}\right)\\
    \bvY_t \mid (X_t = 2) &\sim C_\text{Frank}\left(\sN(\mu_{2,1}, \sigma^2_{2,1}), \sN(\mu_{2,2}, \sigma^2_{2,2}) \mid \theta_{2}\right).\\
\end{align*}

\begin{table}[ht]
\centering
\begin{tabular}{r|rrrrr}
  \toprule
 State & Frank & Clayton & Gumbel & Joe & Gauss \\ 
  \midrule
1 & 0.356 & \textbf{0.255} & 0.423 & 0.770 & 0.345 \\ 
  2 & \textbf{0.018} & 0.433 & 0.038 & 0.206 & 0.045 \\ 
   \bottomrule
\end{tabular}
\caption{Cramér-von Mises test statistics based on pseudo-observations computed from unoccupied (Row 1) and occupied (Row 2) subsets.}
\label{tab:GoFcopulas}
\end{table}

\begin{figure}[ht]
    \centering
    \includegraphics[scale=0.15]{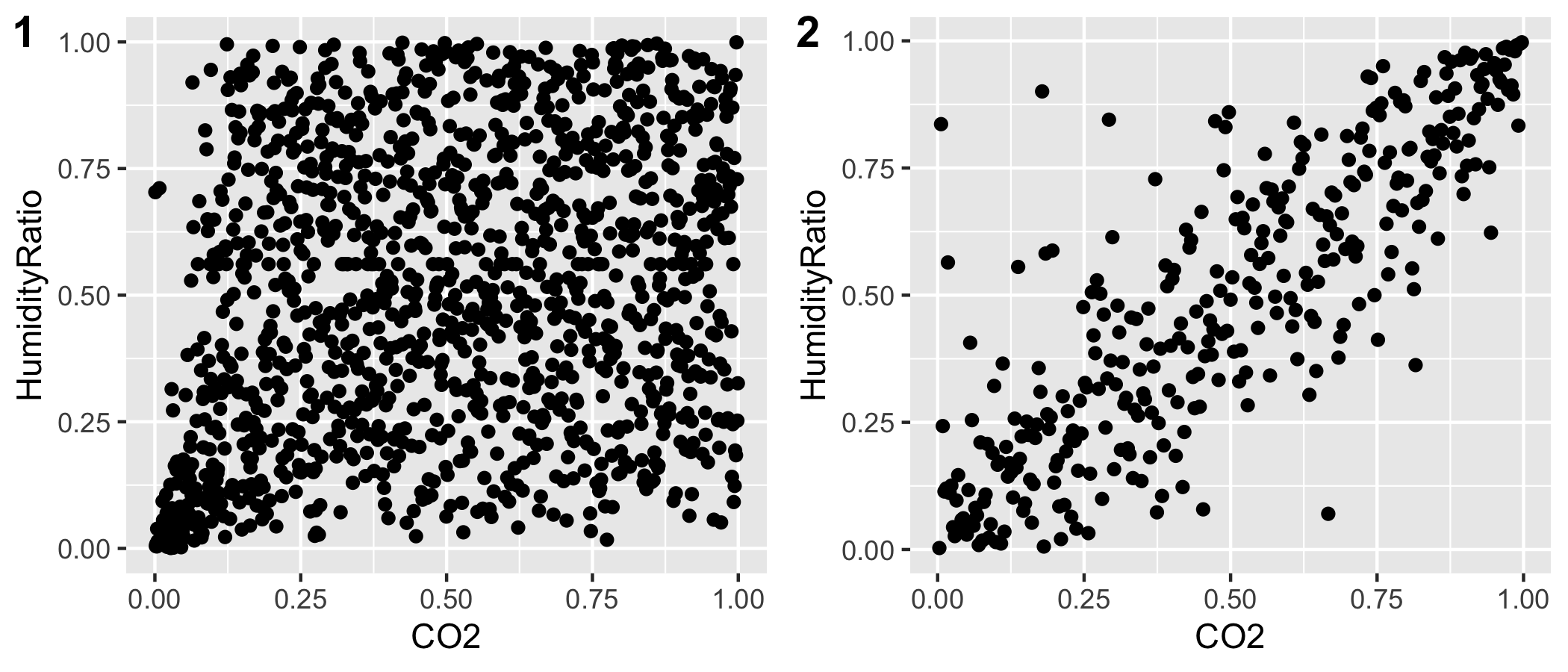}
    \caption{Pseudo-observations computed from unoccupied (Panel 1) and occupied (Panel 2) subsets.}
    \label{fig:pobs_occupancy}
\end{figure}

The parameter estimates produced by the EIFM algorithm (including those for the initial distribution of the Markov chain $\bvpi$ and the chain's transition probability matrix $\bvGamma$) are given in \autoref{tab:occupancyests}, along with 95\% bootstrap confidence intervals produced using the method described in \autoref{sub:estimator}. Local decoding was used to classify all three datasets, and the accuracy of these classifications was assessed via the zero-one loss. To compare the predictive performance of the copula-within-HMM model with that of a more basic model, we repeated the experiment using an independence copula. The performance of both models is shown in \autoref{tab:occupancyZO}. The copula model performs well, correctly classifying almost 90\% of the states in the training dataset. While the independence model also performs well on the training dataset (similar to results in \cite{candanedo2017methodology}), it performs noticeably worse for classifying the states in the test dataset.

\begin{table}[ht]
\centering
\begin{tabular}{r|rrr}
  \toprule
 Parameter & Lower & Estimate & Upper \\ 
  \midrule
$\pi_{1}$ & -0.336 & $2.243 \times 10^{-8}$ & 0.336 \\ 
  $\pi_{2}$ & 0.664 &     1 & 1.336 \\ 
  $\gamma_{1,1}$ & 0.9882 & 0.9935 & 0.9988 \\ 
  $\gamma_{1,2}$ & 0.002489 & 0.01562 & 0.02876 \\ 
  $\gamma_{2,1}$ & 0.001206 & 0.006509 & 0.01181 \\ 
  $\gamma_{2,2}$ & 0.9712 & 0.9844 & 0.9975 \\ 
  $\mu_{1,1}$ & -0.4472 & -0.2685 & -0.08979 \\ 
  $\mu_{1,2}$ & $-8.785 \times 10^{-7}$ & $-3.353 \times 10^{-7}$ & $2.078 \times 10^{-7}$ \\ 
  $\mu_{2,1}$ & -1.876 & 0.873 & 3.622 \\ 
  $\mu_{2,2}$ & $-6.644 \times 10^{-7}$ & $2.562 \times 10^{-6}$ & $5.788 \times 10^{-6}$ \\ 
  $\sigma_{1,1}$ & 2.732 & 2.854 & 2.976 \\ 
  $\sigma_{1,2}$ & $9.482 \times 10^{-6}$ & $9.894 \times 10^{-6}$ & $1.031 \times 10^{-5}$ \\ 
  $\sigma_{2,1}$ & 30.01 & 31.89 & 33.77 \\ 
  $\sigma_{2,2}$ & $3.418 \times 10^{-5}$ & $3.613 \times 10^{-5}$ & $3.808 \times 10^{-5}$ \\ 
  $\theta_{1}$ & -0.2687 & 0.1209 & 0.5106 \\ 
  $\theta_{2}$ & 8.918 & 9.776 & 10.63 \\ 
   \bottomrule
\end{tabular}
\caption{Parameter estimates and 95\% bootstrap confidence intervals for the copula-within-HMM model fit to the training data.}
\label{tab:occupancyests}
\end{table}

\begin{table}[ht]
\centering
\begin{tabular}{r|rrr}
  \toprule
Copula Model & Train & Test 1 & Test 2 \\ 
  \midrule
  Independence & 0.8950 &0.8459 & 0.6795 \\ 
Clayton/Frank & 0.8986 & 0.8515 & 0.6959 \\ 
   \bottomrule
\end{tabular}
\caption{Overall state classification accuracy for the training dataset and the two test datasets using either the independence or Clayton/Frank copula.}
\label{tab:occupancyZO}
\end{table}

\section{Discussion}

This paper proposes a copula-based approach for integration of information from multivariate  observations in a HMM setting. The dependence between the components of the vector is captured using a copula, and we demonstrate that ignoring the dependence can deteriorate the estimation capacity of the statistical model. The computational challenges encountered when fitting the copula HMM model are met using a new iterative algorithm, the EIFM algorithm, which is designed broadly along the ideas of the EM algorithm but differs in essential ways; most importantly, it makes computation possible when other traditional methods fail.  

The EIFM algorithm opens the door to developing other common HMM extensions when considering copula-based state-dependent distributions, most of which will impact the IFM step in the algorithm. One of the most common extensions is to assume a time-inhomogeneous state process such that the dynamics of the transition probability matrix $\boldsymbol{\Gamma}$ depend on time-varying covariates of interest. A related extension is to allow for the state-dependent parameters to vary over time. However, once the copula is introduced, variations of the HMM will usually require specific customization, depending on the application itself; for example, one might want to allow for the marginal distributions or the copula parameters (or both) to vary over time, or for situations in which components of the vector of observations are observed at different times. One might also consider nonparametric estimation of marginals. 

Overall, our aim is to challenge the assumption of conditional contemporaneous independence often made in HMM settings when working with multivariate data, and when assuming longitudinal independence, to challenge the default distributional assumption in the choice of the multivariate Gaussian. We demonstrate the utility of both the Frank and Clayton copulas in an HMM setting and demonstrate that with the aid of labelled data, we can even choose different copula families across states. However, our work here is limited to continuous-valued observation processes and to a bivariate setting. In cases where the dimension of the observation process is greater than two, we will investigate the use of vine copulas \citep{czado,min}. Further, the support of the observation process may not always be the same across dimensions. For example, having a multivariate process where one dimension is discrete-valued, another continuous-valued and a third circular-valued is not uncommon in ecological and environmental applications of HMMs \citep{hodel2022}. Such situations, will require customizing the copula techniques that were developed for mixed variables \citep{cra-sabeti} or discrete ones \citep{genest2014empirical}.

Lastly, in this work we limited our HMM presentation to a discrete-time finite-state framework. However, the EIFM algorithm can be adapted to support inference in other common HMM frameworks such as the continuous-time HMM and hidden semi-Markov model (HSMM). In both of these scenarios, the time at which the state evolves differs in distribution from that of the HMM presented here. Continuous-time HMMs are commonly applied in medical applications where the data collected over time at highly irregular time intervals. With the advent of devices and tests that can collect multivariate observations, the copula-within-HMM framework has the potential to better capture the correlation between various aspects of someone's health profile. The HSMM is a discrete-time framework but with the added choice of selecting appropriate state-dwell  distributions, (i.e., the amount of time spent in a state before switching). In a basic HMM, the state-dwell time distribution is necessarily geometric, while in a HSMM it is common to select Poisson or negative binomial distributions (although semi-parametric forms can also be specified \citep{pohle2022}). In all cases, modifications to the EIFM algorithm will need to be made to the forms presented in both the E-step and the IFM-step.

\section*{Acknowledgments}
RZ was supported by an Ontario Graduate Scholarship. RVC's  and VLB's work have been supported by the Natural Sciences and Engineering Research Council of Canada. RZ thanks Mich\"ael Lalancette and Yanbo Tang for helpful discussions.

\bibliography{References}

\newpage

\appendix
      
\section{Appendix}\label{sec:appendix}

\renewcommand{\log}[1]{\hspace{1pt}\mathrm{log}\left(#1\right)}

\subsection{Regularity Conditions for \autoref{thm:jensensequence}}\label{app:regconditions}

The following definition and assumptions are extracted from \cite{jensen2011asymptotic}.

\begin{definition*}
A set of functions $\sA := \{\bva_t: \Theta \times \sY \times \sX^2 \to \R^n\}_{t \geq 1}$ belongs to the class $C_k$ if there exist functions $\{a_t^0:\sY \to \R\}_{t \geq 1}$, a constant $\delta_0 > 0$, and a constant $K$ such that for all $t \geq 1$,
\[
    \sup_{x_{1:2} \in \sX^2, \, \bveta \in B(\bveta_0, \delta_0)}\left|\bva_t(\bveta, \bvy, x_{1:2})\right| 
    \leq a^0_t(\bvy) 
    \qquad \text{and} \qquad 
    \mathbb{E}_{\bveta_0}\left[a_t^0(\bvY_t)^k \right] \leq K.
\]
Moreover, $\sA$ belongs to the class $C_{k,m}$ if it belongs to $C_k$, and in addition, there exist functions $\{\bar{a}_t:\sY \to \R\}_{t \geq 1}$ a constant $\bar{\delta}_0 > 0$, and a constant $\bar{K}$ such that for all $\bveta \in B(\bveta_0, \bar{\delta}_0)$,
\[
    \left|\bva_t(\bveta, \bvy, x_{1:2}) - \bva_t(\bveta_0, \bvy, x_{1:2})\right| 
    \leq |\bveta - \bveta_0| \cdot \bar{a}_t(\bvy) \text{ for all } x_{1:2} \in \sX^2
    \qquad \text{and} \qquad  
    \mathbb{E}_{\bveta_0}\left[\bar{a}_t(\bvY_t)^m \right] \leq \bar{K}.
\]
\end{definition*} 

Below, we write $\omega_t(\bveta; x_{(t-1):t}, \bvy_t) = \llog{\gamma_{x_{t-1}, x_t} \cdot h_{x_t}(\bvy_t)}$. Then the following assumptions are to be made for \autoref{thm:jensensequence} to hold:

\begin{enumerate}[\bfseries {Assumption }1:]
\item The entries of $\bvGamma_0$ are all strictly positive and $0 < \sumjK \nu_{j,t} \cdot h_j(\bvy_t) < \infty$ for all $\bvy_t \in \R^d$ and all $\bveta \in B(\bveta_0; \delta_0)$. 
\item $\{\bvphi_t(\bveta; x_{(t-1):t}, \bvy_t)\}_{t \geq 1}$ is of class $C_4$, and $\{\frac{\partial}{\partial \eta_r} \omega_t(\bveta; x_{(t-1):t}, \bvy_t) \}_{t \geq 1}$ and $\{ \frac{\partial}{\partial \eta_r} \bvphi_t(\bveta; x_{(t-1):t}, \bvy_t) \}_{t \geq 1}$ are of class $C_4$ and $C_{3,1}$, respectively, for each $r=1,\ldots, p$.
\item There exist constants $K_0 > 0$ and $t_0$ such that for $t > t_0$, 
\[
    \bva^T \bvG_t \bva \geq t K_0 \norm{\bva}^2 
    \qquad \text{for all } 
    \bva \in \R^{p}.
\]
\item There exist constants $c_0 > 0$ and $t_0$ such that for $t > t_0$, the eigenvalues of $\bvH_t^T \bvH_t$ are bounded below by $c_0$. 
\end{enumerate}

\subsection{Proofs}\label{app:proofs}

In the following, we write $X_{-t} = (X_1,\ldots, X_{t-1}, X_{t+1},\ldots, X_T)$ and $X_{-\{t-1,t\}} = (X_1,\ldots, X_{t-2}, X_{t+1},\ldots, X_T)$.

\begin{lemma}\label{lemma:CP1}
Under the model described by \eqref{eq:state-dependentDF}, we have 
\begin{equation}\label{eq:lemma1}
    \parPrb{\bveta}{X_t = k \mid \bvY\oneT = \bvy\oneT} 
    = \frac{h_k(\bvy_t) \cdot \sum_{x_{-t} \in \sX^{T-1}} \kappa(x_{-t}, \bvy_{-t}) \cdot \gamma_{x_{t-1},k} \cdot \gamma_{k, x_{t+1}}}{\sumjK h_j(\bvy_t) \cdot \sum_{x_{-t} \in \sX^{T-1}} \kappa(x_{-t}, \bvy_{-t}) \cdot \gamma_{x_{t-1},j} \cdot \gamma_{j, x_{t+1}}}.
\end{equation}
where
\[
    \kappa(x_{-t}, \bvy_{-t}) = \pi_{x_1} \prod_{s \neq t} h_{x_s}(\bvy_s) \cdot \prod_{s \neq t,t+1} \gamma_{x_{s-1}, x_s}.
\]
In particular, if the system constitutes a finite mixture model with $\parPrb{\bveta}{X_t = j} = \nu_{j}$, we have
\begin{equation}\label{eq:lemma1FMM}
    \parPrb{\bveta}{X_t = k \mid \bvY\oneT = \bvy\oneT} 
    = \frac{\nu_k \cdot h_k(\bvy_t)}{\sumjK \nu_j \cdot h_j(\bvy_t)}.
\end{equation}
\end{lemma}

\begin{proof}
Employing a standard abuse of notation in which densities are replaced by probability measures where appropriate, Bayes' rule gives us 
\begin{equation}
    \parPrb{\bveta}{X_t = k \mid \bvY\oneT = \bvy\oneT} = \frac{\parPrb{\bveta}{\bvY\oneT = \bvy\oneT, X_t = k}}{\sumjK \parPrb{\bveta}{\bvY\oneT = \bvy\oneT, X_t = j}}.\label{eq:bayesbasic}
\end{equation}

Now, for any $j \in \sX$, the law of total probability and the conditional independence of the observations $\bvY\oneT$ respectively give 
\begin{align}
    \parPrb{\bveta}{\bvY\oneT = \bvy\oneT, X_t = j}
    &= \sum_{x_{-t} \in \sX^{T-1}} \parPrb{\bveta}{\bvY\oneT = \bvy\oneT, X_t = j, \X_{-t} = x_{-t}}\nonumber\\
    &= \sum_{x_{-t} \in \sX^{T-1}} \parPrb{\bveta}{\bvY\oneT = \bvy\oneT \mid X_t = j, X_{-t} = x_{-t}} \cdot \parPrb{\bveta}{X_t = j, X_{-t} = x_{-t}}\nonumber\\
    &= \sum_{x_{-t} \in \sX^{T-1}} h_j(\bvy_t) \left( \prod_{s \neq t} h_{x_s}(\bvy_s) \right) \cdot \parPrb{\bveta}{X_t = j, X_{-t} = x_{-t}}\nonumber\\
    &= \sum_{x_{-t} \in \sX^{T-1}} h_j(\bvy_t) \left( \prod_{s \neq t} h_{x_s}(\bvy_s) \right)\cdot \pi_{x_1} \prod_{s=2}^{t-1} \gamma_{x_{s-1}, x_s} \cdot \gamma_{x_{t-1}, j} \cdot \gamma_{j, x_{t+1}} \cdot \prod_{s=t+2}^T \gamma_{x_{s-1}, x_s}\nonumber\\
    &= h_j(\bvy_t) \sum_{x_{-t} \in \sX^{T-1}} \kappa(x_{-t}, \bvy_{-t}) \cdot \gamma_{x_{t-1},j} \cdot \gamma_{j, x_{t+1}}. \label{eq:CPconstant}
\end{align}
Inserting \eqref{eq:CPconstant} into both numerator and denominator of \eqref{eq:bayesbasic} yields \eqref{eq:lemma1}. 

If the system constitutes a finite mixture model, then the Markov structure of the underlying state sequence reduces to independence; that is, $X_{t+1} \indep X_{t}$ for all $t \in \{1, \ldots, T\}$, and in particular, the rows of $\bvGamma$ are identical, with $\gamma_{j,k} = \nu_{k}$ for all $j,k \in \sX$. Then from \eqref{eq:CPconstant} we have that
\begin{align}
    \parPrb{\bveta}{\bvY\oneT = \bvy\oneT, X_t = j}
    &= h_j(\bvy_t) \sum_{x_{-t} \in \sX^{T-1}} \kappa(x_{-t}, \bvy_{-t}) \cdot \gamma_{x_{t-1},j} \cdot \gamma_{j, x_{t+1}}\nonumber\\
    &= h_j(\bvy_t) \sum_{x_{-t} \in \sX^{T-1}} \kappa(x_{-t}, \bvy_{-t}) \cdot \nu_j \cdot \nu_{x_{t+1}}\nonumber\\
    &= C \cdot \nu_j \cdot h_j(\bvy_t),  \label{eq:CPconstant2}
\end{align}
where $C = \sum_{x_{-t}} \kappa(x_{-t}, \bvy_{-t}) \cdot \nu_{x_{t+1}}$ is constant with respect to both $j$ and $\bvy_t$. Inserting \eqref{eq:CPconstant2} into both numerator and denominator of \eqref{eq:bayesbasic} yields \eqref{eq:lemma1FMM} upon factoring out $C$. \end{proof}

\begin{theorem}\label{thm:lossfunctions}
Let $\nu_{t,k} = \Prb{X_t = k} = [\bvpi \bvGamma^t]_k$. The expected zero-one loss of the classifications made by local decoding is given by
\begin{equation}\label{eq:postprobHMM}
\ell_{01}(\bveta) 
= \frac{1}{T}\sumtT \sumkK \nu_{t,k} \cdot \parPrb{\bveta}{\left.\frac{h_k(\bvY_t) \cdot \sum_{x_{-t} \in \sX^{T-1}} \kappa(x_{-t}, \bvY_{-t}) \cdot \gamma_{x_{t-1},k} \cdot \gamma_{k, x_{t+1}}}{\max_{j \neq k}\left\{ h_j(\bvY_t) \cdot \sum_{x_{-t} \in \sX^{T-1}} \kappa(x_{-t}, \bvY_{-t}) \cdot \gamma_{x_{t-1},j} \cdot \gamma_{j, x_{t+1}} \right\}} < 1 \right| X_t = k}
\end{equation}
where $\kappa(x_{-t}, \bvy_{-t}) = \pi_{x_1} \prod_{s \neq t} h_{x_s}(\bvy_s) \cdot \prod_{s \neq t,t+1} \gamma_{x_{s-1}, x_s}$. Moreover, if the HMM constitutes a finite mixture model, then \eqref{eq:ZOELapp} reduces to
\begin{equation}\label{eq:postprobFMM}
\ell_{01}(\bveta) 
= \frac{1}{T}\sumtT \sumkK \nu_{t,k} \int_{\R^d} \mathbbm{1}\left\{\frac{\nu_{t,k} \cdot h_k(\bvy_t)}{\max_{j \neq k} \nu_{t,j} \cdot h_j(\bvy_t)} < 1\right\} \dif H_k(\bvy_t).
\end{equation}
\end{theorem} 

\begin{proof}
Write $\hX_t$ for the prediction made at time $t$. The local decoding algorithm selects for $\hX_t$ the state $k$ which maximizes the conditional probability $\Prb{X_t = k \mid \bvY\oneT = \bvy\oneT}$, which by \autoref{lemma:CP1} is equivalent to 
\begin{align*}
    \hX_t &= \argmax{k \in \sX} \left(\frac{h_k(\bvy_t) \cdot \sum_{x_{-t} \in \sX^{T-1}} \kappa(x_{-t}, \bvy_{-t}) \cdot \gamma_{x_{t-1},k} \cdot \gamma_{k, x_{t+1}}}{\sumjK h_j(\bvy_t) \cdot \sum_{x_{-t} \in \sX^{T-1}} \kappa(x_{-t}, \bvy_{-t}) \cdot \gamma_{x_{t-1},j} \cdot \gamma_{j, x_{t+1}}}\right)\\
    &= \argmax{k \in \sX} \left( h_k(\bvy_t)  \sum_{x_{-t} \in \sX^{T-1}} \kappa(x_{-t}, \bvy_{-t}) \cdot \gamma_{x_{t-1},k} \cdot \gamma_{k, x_{t+1}} \right).
\end{align*}
Thus, foregoing the possibility of ties (which almost surely do not occur), we have that
\begin{align}
    &\phantom{{}\iff{}}\hX_t \neq k \nonumber \\
    & \iff h_k(\bvy_t)  \sum_{x_{-t} \in \sX^{T-1}} \kappa(x_{-t}, \bvy_{-t}) \cdot \gamma_{x_{t-1},k} \cdot \gamma_{k, x_{t+1}} \neq \argmax{j \neq k}\left\{ h_j(\bvy_t)  \sum_{x_{-t} \in \sX^{T-1}} \kappa(x_{-t}, \bvy_{-t}) \cdot \gamma_{x_{t-1},j} \cdot \gamma_{j, x_{t+1}} \right\}\nonumber \\
    & \iff h_k(\bvy_t)  \sum_{x_{-t} \in \sX^{T-1}} \kappa(x_{-t}, \bvy_{-t}) \cdot \gamma_{x_{t-1},k} \cdot \gamma_{k, x_{t+1}} < \max_{j \neq k}\left\{ h_j(\bvy_t)  \sum_{x_{-t} \in \sX^{T-1}} \kappa(x_{-t}, \bvy_{-t}) \cdot \gamma_{x_{t-1},j} \cdot \gamma_{j, x_{t+1}} \right\}\nonumber \\
    & \iff \frac{h_k(\bvy_t) \cdot \sum_{x_{-t} \in \sX^{T-1}} \kappa(x_{-t}, \bvy_{-t}) \cdot \gamma_{x_{t-1},k} \cdot \gamma_{k, x_{t+1}}}{\max_{j \neq k}\left\{ h_j(\bvy_t) \cdot \sum_{x_{-t} \in \sX^{T-1}} \kappa(x_{-t}, \bvy_{-t}) \cdot \gamma_{x_{t-1},j} \cdot \gamma_{j, x_{t+1}} \right\}} < 1. \label{eq:fracl1}
\end{align}

Taking expectations with respect to $(\bvY\oneT, X\oneT)$, the expected zero-one loss satisfies
\begin{align*}
    \ell_{0,1}(\bveta) &= \mathbb{E}\left[\frac{1}{T}\sumtT \onee{\hX_t \neq X_t} \right]\\
    &= \frac{1}{T}\sumtT \parPrb{\bveta}{\hX_t \neq X_t}\\
    &= \frac{1}{T}\sumtT \sumkK \nu_{t,k}\cdot \parPrb{\bveta}{\hX_t \neq  k \mid X_t = k}\\
    &= \frac{1}{T}\sumtT \sumkK \nu_{t,j} \cdot \parPrb{\bveta}{\left.\frac{h_k(\bvY_t) \cdot \sum_{x_{-t} \in \sX^{T-1}} \kappa(x_{-t}, \bvY_{-t}) \cdot \gamma_{x_{t-1},k} \cdot \gamma_{k, x_{t+1}}}{\max_{j \neq k}\left\{ h_j(\bvY_t) \cdot \sum_{x_{-t} \in \sX^{T-1}} \kappa(x_{-t}, \bvY_{-t}) \cdot \gamma_{x_{t-1},j} \cdot \gamma_{j, x_{t+1}} \right\}} < 1 \right| X_t = k}.
\end{align*} 
If the system constitutes a finite mixture model, then \eqref{eq:fracl1} reduces to 
\[
    \frac{\nu_{t,k} \cdot h_k(\bvy_t)}{\max_{j \neq k} \nu_{t,j} \cdot h_k(\bvy_t)} < 1
\]
by \eqref{eq:postprobFMM}, and the same derivation leads to 
\[
    \ell_{01}(\bveta) 
    = \frac{1}{T}\sumtT \sumkK \nu_{t,k} \int_{\R^d} \mathbbm{1}\left\{\frac{\nu_{t,k} \cdot h_k(\bvy_t)}{\max_{j \neq k} \nu_{t,j} \cdot h_j(\bvy_t)} < 1\right\} \dif H_k(\bvy_t)
\]
as desired. \end{proof}

\begin{proposition}
Suppose that $\parPrb{\bveta}{X_t = k \mid X_{t-1}, X_{t+1}} \geq \parPrb{\bveta}{X_t = j \mid X_{t-1}, X_{t+1}}$ for all $j \in \sX$. Then
\begin{align*}
&\phantom{{}={}}\parPrb{\bveta}{\left.\frac{h_k(\bvY_t) \cdot \sum_{x_{-t} \in \sX^{T-1}} \kappa(x_{-t}, \bvY_{-t}) \cdot \gamma_{X_{t-1},k} \cdot \gamma_{k, X_{t+1}}}{\max_{j \neq k}\left\{ h_j(\bvY_t) \cdot \sum_{x_{-t} \in \sX^{T-1}} \kappa(x_{-t}, \bvY_{-t}) \cdot \gamma_{X_{t-1},j} \cdot \gamma_{j, X_{t+1}} \right\}} < 1 \right| X_t = k} \\
&\leq \int_{\R^d} \mathbbm{1}\left\{\frac{h_k(\bvy_t)}{\max_{j \neq k} h_j(\bvy_t)} < 1\right\} \dif H_k(\bvy_t).
\end{align*}
\end{proposition}

\begin{proof}

To begin with, note that for any $x_{t-1}, x_{t+1} \in \sX$, 
\begin{align}
    &\phantom{{}={}}\parPrb{\bveta}{X_t = k \mid X_{t-1} = x_{t-1}, X_{t+1} = x_{t+1}} \nonumber \\
    &= \frac{\parPrb{\bveta}{X_{t-1} = x_{t-1}, X_t = k, X_{t+1} = x_{t+1}}}{\parPrb{\bveta}{X_{t-1} = x_{t-1}, X_{t+1} = x_{t+1}}} \label{eq:MCcond}\\
    &= \frac{\parPrb{\bveta}{X_t = k, X_{t+1} = x_{t+1} \mid X_{t-1} = x_{t-1}} \cdot \parPrb{\bveta}{X_{t-1} = x_{t-1}}}{\parPrb{\bveta}{X_{t-1} = x_{t-1}, X_{t+1} = x_{t+1}}}\nonumber \\
    &= \frac{\parPrb{\bveta}{X_t = k \mid X_{t-1} = x_{t-1}} \cdot \parPrb{\bveta}{X_{t+1} = x_{t+1} \mid X_t = k}\cdot \parPrb{\bveta}{X_{t-1} = x_{t-1}}}{\parPrb{\bveta}{X_{t-1} = x_{t-1}, X_{t+1} = x_{t+1}}}\nonumber \\
    &= \frac{\gamma_{x_{t-1},k} \cdot \gamma_{k, x_{t+1}}}{\parPrb{\bveta}{X_{t+1} = x_{t+1} \mid X_{t-1} = x_{t-1}}}, \nonumber
\end{align}
and therefore the Markov chain condition $\parPrb{\bveta}{X_t = k \mid X_{t-1}, X_{t+1}} \geq \parPrb{\bveta}{X_t = j \mid X_{t-1}, X_{t+1}}$ for all $j \in \sX$ is equivalent to
\begin{equation}\label{eq:01preineq}
    \gamma_{x_{t-1},k} \cdot \gamma_{k, x_{t+1}} \geq \max_{j \neq k} \left\{ \gamma_{x_{t-1},j} \cdot \gamma_{j, x_{t+1}} \right\}
\end{equation}
for all $x_{t-1}, x_{t+1} \in \sX$. Thus, for any $\bvy\oneT \in \R^{d \times T}$, we have that
\begin{align*}
    &\phantom{{}\iff{}} h_k(\bvy_t) \geq \max_{j \neq k} h_j(\bvy_t) \\
    &\iff h_k(\bvy_t) \sum_{x_{-t} \in \sX^{T-1}} \kappa(x_{-t}, \bvy_{-t}) \geq \max_{j \neq k} \left\{ h_j(\bvy_t) \sum_{x_{-t} \in \sX^{T-1}} \kappa(x_{-t}, \bvy_{-t}) \right\}\\
    &\implies h_k(\bvy_t) \sum_{x_{-t} \in \sX^{T-1}} \kappa(x_{-t}, \bvy_{-t}) \cdot \gamma_{x_{t-1}, k} \cdot \gamma_{k, x_{t+1}} \geq \max_{j \neq k} \left\{ h_j(\bvy_t) \sum_{x_{-t} \in \sX^{T-1}} \kappa(x_{-t}, \bvy_{-t}) \right\} \cdot \max_{j \neq k} \left\{\gamma_{x_{t-1}, j} \cdot \gamma_{j, x_{t+1}}\right\}\\
    &\implies h_k(\bvy_t) \sum_{x_{-t} \in \sX^{T-1}} \kappa(x_{-t}, \bvy_{-t}) \cdot \gamma_{x_{t-1}, k} \cdot \gamma_{k, x_{t+1}} \geq \max_{j \neq k} \left\{ h_j(\bvy_t) \sum_{x_{-t} \in \sX^{T-1}} \kappa(x_{-t}, \bvy_{-t}) \cdot \gamma_{x_{t-1}, j} \cdot \gamma_{j, x_{t+1}}\right\},
\end{align*}
where we have used \eqref{eq:01preineq} and the standard fact that $\sup_x f(x) \cdot \sup_x g(x) \geq \sup_x \left( f(x) \cdot g(x) \right)$ when $f,g > 0$. It follows that as subsets of $\R^{d \times T}$,
\begin{equation*}
    \left\{\bvy\oneT: \frac{h_k(\bvy_t)}{\max_{j \neq k} h_j(\bvy_t)} \geq 1\right\}
    \subseteq \left\{\bvy\oneT: \frac{h_k(\bvy_t) \cdot \sum_{x_{-t} \in \sX^{T-1}} \kappa(x_{-t}, \bvy_{-t}) \cdot \gamma_{x_{t-1}, k} \cdot \gamma_{k, x_{t+1}}}{\max_{j \neq k} \left\{ h_j(\bvy_t) \cdot \sum_{x_{-t} \in \sX^{T-1}} \kappa(x_{-t}, \bvy_{-t}) \gamma_{x_{t-1}, j} \cdot \gamma_{j, x_{t+1}}\right\}} \geq 1\right\},
\end{equation*}
or equivalently,
\begin{equation*}
    \left\{\bvy\oneT: \frac{h_k(\bvy_t) \cdot \sum_{x_{-t} \in \sX^{T-1}} \kappa(x_{-t}, \bvy_{-t}) \cdot \gamma_{x_{t-1}, k} \cdot \gamma_{k, x_{t+1}}}{\max_{j \neq k} \left\{ h_j(\bvy_t) \cdot \sum_{x_{-t} \in \sX^{T-1}} \kappa(x_{-t}, \bvy_{-t}) \gamma_{x_{t-1}, j} \cdot \gamma_{j, x_{t+1}}\right\}} < 1\right\} \subseteq \left\{\bvy\oneT: \frac{h_k(\bvy_t)}{\max_{j \neq k} h_j(\bvy_t)} < 1\right\}.
\end{equation*}
Thus,
\begin{align*}
    &\phantom{{}={}}\parPrb{\bveta}{\left.\frac{h_k(\bvY_t) \cdot \sum_{x_{-t} \in \sX^{T-1}} \kappa(x_{-t}, \bvY_{-t}) \cdot \gamma_{x_{t-1},k} \cdot \gamma_{k, x_{t+1}}}{\max_{j \neq k}\left\{ h_j(\bvY_t) \cdot \sum_{x_{-t} \in \sX^{T-1}} \kappa(x_{-t}, \bvY_{-t}) \cdot \gamma_{x_{t-1},j} \cdot \gamma_{j, x_{t+1}} \right\}} < 1 \right| X_t = k} \\
    &\leq \parPrb{\bveta}{\left.\frac{h_k(\bvY_t)}{\max_{j \neq k} h_j(\bvY_t)} < 1 \right| X_t = k} \\
    &= \int_{\R^d} \mathbbm{1}\left\{\frac{h_k(\bvy_t)}{\max_{j \neq k} h_j(\bvy_t)} < 1\right\} \dif H_k(\bvy_t).
\end{align*}
\end{proof}

\begin{proposition}\label{prop:COapp}
Let $d=2$ and fix $k \in \sX$. As either $|\tau_k| \to 1$ or $|\rho_k| \to 1$ while the other state-dependent parameters stay fixed, we have 
\[
    \int_{\R^d} \mathbbm{1}\left\{\frac{\omega_{k} \cdot h_k(\bvy)}{\max_{j \neq k} \omega_{j} \cdot h_j(\bvy)} < 1\right\} \dif H_k(\bvy) \to 0
\] 

for any $(\omega_1, \ldots, \omega_K) \in \sS^{K-1}$ with $\omega_k > 0$.
\end{proposition}

\begin{proof}
For ease of notation, we omit the subscript $t$ in what follows. First, note that by Theorem 3.6 of \cite{embrechts2001modelling}, $\rho_k = 1$ if and only if $\tau_k = 1$, either of which are equivalent to $C_k$ being the comonotonicity copula. Similarly, $\rho_k = -1$ if and only if $\tau_k = -1$, corresponding to the countermonotonicity copula.

Let $\epsilon > 0$, and start by assuming $\rho_k \to 1$. Then
\begin{align}
    \int_{\R^d} \onee{\frac{\omega_{k} \cdot h_k(\bvy)}{\max_{j \neq k} \omega_{j} \cdot h_j(\bvy)} < 1 }\dif H_k(\bvy)
    &=  \int_{\R^2} \onee{\max_{j \neq k} \omega_{j} \cdot h_j(\bvy) > \omega_{k} \cdot h_k(\bvy) }\dif H_k(\bvy) \nonumber\\
    &=  \int_{\R^2} \onee{\bigvee_{j \neq k} \omega_{j} \cdot h_j(\bvy) > \omega_{k} \cdot h_k(\bvy) }\dif H_k(\bvy) \nonumber\\
    &\leq \sum_{j \neq k}\int_{\R^2} \onee{\omega_{j} \cdot h_j(\bvy) >  \omega_{k} \cdot h_k(\bvy)} \dif H_k(\bvy),\label{eq:boole}
\end{align}
where we have used Boole's inequality in the last line. Consider now the integral 
\begin{equation}\label{eq:integraljgeqk}
    \int_{\R^2} \onee{\omega_{j} \cdot h_j(\bvy) >  \omega_{k} \cdot h_k(\bvy)} \dif H_k(\bvy)
\end{equation}
inside the summation of \eqref{eq:boole}. If $\omega_j = 0$ then \eqref{eq:integraljgeqk} is zero and contributes nothing to the sum, so we may assume $\omega_j > 0$ and let $\bar{\omega}_j = \omega_j/\omega_k$. Further, let $d_j(\bvu) = c_j(F_{j,1}(F_{k,1}^{-1}(u_1)), F_{j,2}(F_{k,2}^{-1}(u_2)) \cdot \prod_{h=1}^2 f_{j,h}(F_{k,h}^{-1}(u_h))/f_{k,h}(F_{k,h}^{-1}(u_h))$ and $\bvu = \left(F_{k,1}(y_1), F_{k,2}(y_2)\right)$ so that by a change of variables, \eqref{eq:integraljgeqk} becomes
\begin{align*}
\int_{\R^2} \onee{\bar{\omega}_j h_j(\bvy) > h_k(\bvy)} h_k(\bvy) \dif \bvy &= \int_{[0,1]^2} \onee{\bar{\omega}_j d_j(\bvu) > c_k(\bvu)} c_k(\bvu)\dif \bvu\\
&= \int_{B_\delta} \onee{\bar{\omega}_j d_j(\bvu) > c_k(\bvu)} c_k(\bvu)\dif \bvu +\int_{B^c_\delta} \onee{\bar{\omega}_j d_j(\bvu) > c_k(\bvu)} c_k(\bvu)\dif \bvu,
\end{align*}
where $B_\delta = \{\bvu \in [0,1]^2: |u_1-u_2|\leq \delta\}$, for some $\delta > 0$ to be specified later. Now for $(U_1,U_2) \sim C_k$, we have, by Markov's inequality,
\begin{equation}\label{eq:ZOproof1}
     \int_{B^c_\delta} \onee{\bar{\omega}_j d_j(\bvu) > c_k(\bvu)} c_k(\bvu)\dif \bvu 
     \leq \int_{B^c_\delta} c_k(\bvu) \dif \bvu
     = \Prb{|U_1 - U_2| > \delta}
     \leq \frac{\E{(U_1-U_2)^2}}{\delta^2}
     = \frac{1 - \text{Corr}(U_1,U_2)}{6\delta^2},
\end{equation}
On the other hand,
\begin{align}
    \int_{B_\delta} \onee{\bar{\omega}_j d_j(\bvu) > c_k(\bvu)} c_k(\bvu)\dif \bvu &< \int_{B_\delta} \onee{\bar{\omega}_j d_j(\bvu) > c_k(\bvu)} \bar{\omega}_j d_j(\bvu)\dif \bvu \nonumber \\
    &\leq \int_{B_\delta} \bar{\omega}_j d_ j(\bvu)\dif \bvu \nonumber \\
    &= \bar{\omega}_j \int_{B_\delta} c_j(F_{j,1}(F_{k,1}^{-1}(u_1)), F_{j,2}(F_{k,2}^{-1}(u_2)) \cdot \prod_{h=1}^2 f_{j,h}(F_{k,h}^{-1}(u_h))/f_{k,h}(F_{k,h}^{-1}(u_h)) \dif \bvu \label{eq:integralBdelta}.
\end{align}
Let $B'_\delta = \{\bvw \in \R^2: |F_{k,1}(w_1) - F_{k,2}(w_2)| \leq \delta\}$. Making the substitution $\bvw =(F_{k,1}^{-1}(u_1), F_{k,2}^{-1}(u_2))$, \eqref{eq:integralBdelta} becomes
\begin{align*}
    &\phantom{{}={}}\bar{\omega}_j \int_{B'_\delta} c_j(F_{j,1}(w_1),F_{j,2}(w_2)) \cdot \prod_{h=1}^2 \left(f_{j,h}(w_h)/f_{k,h}(w_h)\right) \cdot \prod_{h=1}^2 f_{k,h}(w_h) \dif \bvw \\
    &= \bar{\omega}_j \int_{B'_\delta} c_j(F_{j,1}(w_1),F_{j,2}(w_2)) \cdot \prod_{h=1}^2 f_{j,h}(w_h) \dif \bvw \\
    &= \bar{\omega}_j \int_{B'_\delta} h_j(\bvw) \dif \bvw \\
    &= \bar{\omega}_j \cdot \mu_j(B'_\delta),
\end{align*}
where $\mu_j$ is the measure on $\R^2$ induced by $H_j$. Since $\mu_j$ is absolutely continuous with respect to Lebesgue measure $\Lambda(\cdot)$ on $\R^2$, there exists some $\delta_\epsilon$ such that whenever $\Lambda(B'_\delta) < \delta_\epsilon$, we have $\mu_j(B'_\delta) < \epsilon/((K-1)\bar{\omega}_j)$. Since we have $\Lambda(B'_\delta) \to 0$ as $\delta \to 0$, there must exist some $\delta^*$ small enough that $\Lambda(B'_{\delta^*}) < \delta_\epsilon$. Making the choice $\delta = \delta^*$ thus gives 
\begin{equation}\label{eq:ZOproof2}
    \int_{B_\delta} \onee{\bar{\omega}_j d_j(\bvu) > c_k(\bvu)} c_k(\bvu)\dif \bvu < \frac{\epsilon}{K-1}. 
\end{equation}
Now, $\rho_k = \text{Corr}(U_1,U_2)$ because $H_k$ has continuous margins \citep{nelsen2007introduction}, and so upon combining \eqref{eq:ZOproof1} and \eqref{eq:ZOproof2}, we obtain 
\begin{equation}\label{eq:ZOproof4}
    \limsup_{\rho_k \to 1}\int_{\R^2} \onee{\bar{\omega}_j h_j(\bvy) > h_k(\bvy)} h_k(\bvy) \dif \bvy < \limsup_{\rho_k \to 1} \left(\frac{1 - \rho_k}{6\delta^2} + \frac{\epsilon}{K-1} \right) \leq \frac{\epsilon}{K-1}.
\end{equation} 
Inserting \eqref{eq:ZOproof4} into \eqref{eq:boole} gives 
\begin{equation}\label{eq:ZOproofconv}
    \limsup_{\rho_k \to 1} \int_{\R^2} \mathbbm{1}\left\{\frac{\omega_{k} \cdot h_k(\bvy)}{\max_{j \neq k} \omega_{j} \cdot h_j(\bvy)} < 1\right\} \dif H_k(\bvy) \leq \epsilon.
\end{equation}
Finally, if instead we have $\rho_k \to -1$, we require the numerator of \eqref{eq:ZOproof1} to be $1 + \text{Corr}(U_1,U_2)$, and this only needs the modification $B_\delta = \{\bvu \in [0,1]^2: |u_1+u_2-1| \leq \delta\}$ and similarly for $B'_\delta$, aside from which the proof is identical. In either case, \eqref{eq:ZOproofconv} can be made arbitrarily close to 0, which completes the proof. \end{proof}

\begin{corollary} \label{cor:radsym01}
Let $d=2$ and $K=2$. Suppose $\mathcal{C} = \{C_{\theta}(\cdot,\cdot):\theta \in \Theta\}$ is a parametric family of radially symmetric bivariate copulas satisfying $C_{-\theta}(u,v) = u - C_{\theta}(u,1-v) = v - C_{\theta}(1-u,v)$. If our HMM is a finite mixture model with $C_1 = C_\theta$ and $C_2 = C_{-\theta}$ with identical marginals and $\Prb{X_t=1} = \Prb{X_t =2} = 1/2$ for all $t \geq 1$, then \begin{equation}\ell_{01}(\theta) = 1 - 2 \cdot C_{\theta}\left(\frac{1}{2},\frac{1}{2}\right).\label{eq:01loss}\end{equation}In particular, for $\theta > 0$ we have $\ell_{01}^\mathrm{Frank}(\theta) = 1 - 2 \llog{\cosh{\frac{\theta}{4}}}/\theta$, $\ell_{01}^\mathrm{Gauss}(\rho) = \cos^{-1}(\rho)/\pi$, and $\ell_{01}^\mathrm{FGM}(\theta) = 1/2 - \theta/8$. 
\end{corollary}

\begin{remark}
In all three cases, $\ell_{01}(0) = 1/2$, which is in line with intuition. While $\ell_{01}^\mathrm{Frank}(\pm \infty) = \ell_{01}^\mathrm{Gauss}(\pm 1) = 0$, we only have $\ell_{01}^\mathrm{FGM}(1) = 3/8$. That $\ell_{01}^\mathrm{FGM}(1) \neq 0$ arises from the fact that the FGM copula can model only weak dependence \citep{nelsen2007introduction}, and thus is not comprehensive. The Plackett family of copulas is radially symmetric, but with parameter space $(0,\infty)$, it does not satisfy the relationship in the statement of the corollary. It does, however, satisfy $C_{1/\theta}(u,v) = u -C_\theta(u,v)$ \citep{nelsen2007introduction}, for which a modification of \eqref{eq:01loss} can be developed.
\end{remark}

\begin{proof}
Since $\nu_{t,k} = 1/2$ for all $t \geq 1$ and $k\in \sX$, using \autoref{thm:lossfunctions}, the zero-one loss in this situation reduces to
\begin{align}
    \ell_{01}(\bveta) 
    &= \frac{1}{2} \sum_{k=1}^2 \int_{\R^2} \onee{\frac{h_k(\bvy)}{\max_{j \neq k} h_j(\bvy)} < 1} \dif H_k(\bvy) \nonumber \\
    &= \frac{1}{2} \left(\int_{\R^2} \onee{h_1(\bvy) < h_2(\bvy)} \dif H_1(\bvy) + \int_{\R^2} \onee{h_2(\bvy) < h_1(\bvy)} \dif H_2(\bvy) \right)\nonumber \\
    &= 1 - \frac{1}{2} \left(\int_{\R^2} \onee{h_1(\bvy) > h_2(\bvy)} \dif H_1(\bvy) + \int_{\R^2} \onee{h_2(\bvy) > h_1(\bvy)} \dif H_2(\bvy) \right)\nonumber \\
    &= 1 - \frac{1}{2} \left(\int_{\R^2} \onee{h_\theta(\bvy) > h_{-\theta}(\bvy)} \dif H_\theta(\bvy) + \int_{\R^2} \onee{h_{-\theta}(\bvy) > h_\theta(\bvy)} \dif H_{-\theta}(\bvy) \right), \label{eq:lossRS}
\end{align}
where we have made the change of notation $h_\theta := h_1$ and $h_{-\theta} := h_2$ for convenience. Now, $h_\theta(\bvy) = c_\theta\left(F_1(y_1),F_2(y_2)\right) f_1(y_1) f_2(y_2)$, so by cancelling out marginal densities and writing $u_h = F_h(y_h) \in [0,1]$ for $h=1,2$, we have
\[
    h_\theta(\bvy) > h_{-\theta}(\bvy) \iff c_{\theta}(u_1,u_2) > c_{-\theta}(u_1,u_2),
\]
and therefore after applying the change of variables $\bvu = (F_1(y_1), F_2(y_2))$, \eqref{eq:lossRS} becomes
\begin{equation}\label{eq:lossRS2}
    1 - \frac{1}{2}\left(\int_0^1\int_0^1\onee{c_\theta(u_1,u_2) > c_{-\theta}(u_1,u_2)} c_{\theta}(u_1,u_2) \dif \bvu + \int_0^1\int_0^1\onee{c_{-\theta}(u_1,u_2) > c_{\theta}(u_1,u_2)} c_{-\theta}(u_1,u_2) \dif \bvu\right).
\end{equation}
Now, the two integrals in \eqref{eq:lossRS2} are \emph{equal}; to show this, we make the temporary change of variables $w = 1-u_2$ and repeatedly use the implication $c_{-\theta}(u_1,u_2) = c_{\theta}(u_1,1-u_2) = c_\theta(1-u_1, u_2)$ to obtain
\begin{align*}
    &\phantom{{}={}}\int_0^1\int_0^1\onee{c_{-\theta}(u_1,u_2) > c_{\theta}(u_1,u_2)} c_{-\theta}(u_1,u_2) \dif \bvu\\
    &= -\int_0^1\int_1^0\onee{c_{-\theta}(u_1,1-w) > c_{\theta}(u_1,1-w)} c_{-\theta}(u_1,1-w) \dif w \dif u_1\\
    &= \int_0^1\int_0^1\onee{c_{\theta}(u_1,w) > c_{-\theta}(u_1,w)} c_{\theta}(u_1,w) \dif w \dif u_1\\
    &= \int_0^1\int_0^1\onee{c_{\theta}(u_1,u_2) > c_{-\theta}(u_1,u_2)} c_{\theta}(u_1,u_2) \dif \bvu
\end{align*}
Moreover, the common value of the two integrals is equal to 
\begin{equation}\label{eq:radialintegral}
    \int_0^1\int_0^1\onee{c_{\theta}(u_1,u_2) > c_{\theta}(u_1,1-u_2)} c_{\theta}(u_1,u_2)  \dif \bvu,
\end{equation}
so 
\[
    \ell_{01}(\bveta) = 1 - \int_0^1\int_0^1\onee{c_{\theta}(u_1,u_2) > c_{\theta}(u_1,1-u_2)} c_{\theta}(u_1,u_2)  \dif \bvu.
\]
It remains to evaluate \eqref{eq:radialintegral}. Let $A = \{\bvu: c_{\theta}(u_1,u_2) > c_{\theta}(u_1,1-u_2)\}$, and partition it as $A = A_1 \cup A_2 \cup A_3 \cup A_4$ with \begin{align*}
    A_1 &= A \cap (0,1/2)^2\\
    A_2 &= A \cap \left((0,1/2) \times (1/2,1)\right)\\
    A_3 &= A \cap \left((1/2,1) \times (0,1/2)\right)\\
    A_4 &= A \cap (1/2,1)^2.
\end{align*}Then \begin{align*}
    A_2 &= \{\bvu : c_\theta(u_1,u_2) > c_\theta(u_1,1-u_2), u_1 \in (0,1/2), u_2 \in (1/2,1)\}\\
    &= \{\bvu : c_\theta(1-u_1,1-u_2) > c_\theta(u_1,1-u_2), u_1 \in (0,1/2),u_2 \in (1/2,1)\}\\
    &= \{\bvu : c_\theta(1-u_1,u_2) > c_\theta(u_1,u_2), u_1 \in (0,1/2),u_2 \in (0,1/2)\}\\
    &= \{\bvu : c_\theta(u_1,u_2) \leq c_\theta(u_1,1-u_2), u_1 \in (0,1/2),u_2 \in (0,1/2)\}\\
    &= A^c \cap (0, 1/2)^2
\end{align*}and so $A_1 \cup A_2 = (0,1/2)^2$. With similar manipulations, we can obtain $A_3 \cup A_4 = (1/2, 1)^2$ and hence $A = (0, 1/2)^2 \cup (1/2, 1)^2$. Thus, \eqref{eq:radialintegral} reduces to \begin{align*}
    &\phantom{{}={}} \int_{0}^{1/2}\int_{0}^{1/2} c_{\theta}(u_1,u_2)  \dif u_2 \dif u_1 + \int_{1/2}^1 \int_{1/2}^1 c_{\theta}(u_1,u_2)  \dif u_2 \dif u_1 \\
    &= \int_{0}^{1/2}\int_{0}^{1/2} c_{\theta}(u_1,u_2)  \dif u_2 \dif u_1 + \int_{1/2}^1 \int_{1/2}^1 c_{\theta}(1-u_1,1-u_2)  \dif u_2 \dif u_1 \\
    &= \int_{0}^{1/2}\int_{0}^{1/2} c_{\theta}(u_1,u_2)  \dif u_2 \dif u_1 + \int_{0}^{1/2}\int_{0}^{1/2} c_{\theta}(u_1,u_2)  \dif u_2 \dif u_1 \\
    &= 2 \cdot C_\theta\left(\frac{1}{2},\frac{1}{2}\right),
\end{align*}which proves the general result. The specific cases for the Frank and FGM families follow from direct calculation, and that for the Gauss copula follows from a well-known result about quadrant probabilities for the bivariate normal distribution (Formula 23.6.19 of \cite{abramowitz1988handbook}). \end{proof}

\begin{remark}
This expression is very close to \emph{Blomqvist's $\beta$}. See Section 5.1.4 of \citep{nelsen2007introduction}.
\end{remark}

\end{document}